\documentclass{article}
\usepackage{hyperref,graphicx}

\usepackage{amsmath,amssymb,amsthm}
\newtheorem{theorem}{Theorem}
\newtheorem{lemma}[theorem]{Lemma}
\newtheorem{proposition}[theorem]{Proposition}
\theoremstyle{plain}

\newcommand{\vect}[1]{\overrightarrow{#1}}

\newcommand{\cq}{\mathcal{Q}}
\newcommand{\cv}{\mathcal{V}}
\newcommand{\bw}{\mathbf{w}}

\DeclareMathOperator{\abov}{Above}
\DeclareMathOperator{\righ}{Right}
\DeclareMathOperator*{\height}{height}
\DeclareMathOperator*{\width}{width}
\DeclareMathOperator*{\bigL}{\mathcal{L}}
\DeclareMathOperator*{\bigC}{\mathcal{C}}
\DeclareMathOperator*{\bigN}{\mathbb{N}}
\DeclareMathOperator*{\bigZ}{\mathbb{Z}}

\DeclareMathOperator{\occ}{Occ}

\DeclareMathOperator*{\dom}{dom}

\title{Comparison of Coverability and Multi-Scale Coverability in
  One and Two Dimensions}
\author{Guilhem Gamard and Gwena{\"{e}}l Richomme}

\date{June 20, 2015}

\begin{document}

\maketitle

\begin{abstract}
  A word is \emph{quasiperiodic} (or \emph{coverable}) if it can be
  covered with occurrences of another finite word, called its
  \emph{quasiperiod}. A word is \emph{multi-scale quasiperiodic} (or
  \emph{multi-scale coverable}) if it has infinitely many different
  quasiperiods. These notions were previously studied in the domains
  of text algorithms and combinatorics of right infinite words.
  
  We extend them to infinite pictures (two-dimensional words). Then we
  compare the regularity properties (uniform recurrence, uniform
  frequencies, topological entropy) of quasiperiodicity with
  multi-scale quasiperiodicity, and we also compare each of them with
  its one-dimensional counterpart.
  
  We also study which properties of quasiperiods enforce properties on
  the quasiperiodic words.
\end{abstract}

\section{Introduction}
\label{sec:intro}

At the beginning of the 1990's, in the area of text algorithms,
Apostolico and Ehrenfeucht introduced the notion of
\emph{quasiperiodicity}~\cite{ApostolicoEhrenfeucht1993TCS}. Their
definition is as follows: ``a string $w$ is quasiperiodic if there is
a second string $u \neq w$ such that every position of $w$ falls
within some occurrence of $u$ in $w$''. The word $w$ is also said to
be $u$-quasiperiodic, and $u$ is called a \emph{quasiperiod} (or a
\emph{cover}) of $w$. For instance, the string:
\begin{equation*}
  ababaabababaababababaababa
\end{equation*}
is $aba$-quasiperiodic and $ababa$-quasiperiodic.

In 2004, Marcus extended this notion to right-infinite words and
observed some basic facts about this new class. He opened several
questions~\cite{Marcus2004BEATCS}, most of them related to Sturmian
words and the subword complexity. First answers were given
in~\cite{LeveRichomme2004BEATCS}. A characterization of right-infinite
quasiperiodic Sturmian words was given in~\cite{LeveRichomme2007TCS}
and extended to episturmian words in~\cite{GlenLeveRichomme2008TCS}.
More details on the complexity function were given
in~\cite{PolleyStaiger2010DFCS, MarcusMonteil2006Arxiv}.

In~\cite{MarcusMonteil2006Arxiv}, Marcus and Monteil showed that
quasiperiodicity is independent from several other classical notions
of symmetry in combinatorics on words. They also introduced a stronger
notion, namely \emph{multi-scale quasiperiodicity}, with better
properties.

Finally, in~\cite{CrochemoreIliopoulosKorda1998Algorithmica}, a
two-dimensional version of quasiperiodicity was introduced. In
particular, a linear-time algorithm computing all square quasiperiods
of a square matrix of letters was given.

\paragraph{Warning}
Note that in some contexts, most notably in the fields of sub-shifts,
symbolic dynamics and tilings, ``quasiperiodic'' means ``uniformly
recurrent''. Since we intend to move towards these areas in the
future, the risk of collision is very high. Hence, from now on, we
refer to quasiperiodic words as \emph{coverable} words; each
quasiperiod is a \emph{cover} (or \emph{covering pattern}).

\medskip\noindent
In~\cite{GamardRichomme2015Lata}, we continued the study of
two-dimensional coverability by generalizing the results
from~\cite{MarcusMonteil2006Arxiv} to infinite pictures. In
particular, we have shown some dependence and independence results
between coverability (resp. multi-scale coverability) and
aperiodicity, uniform recurrence, uniform frequencies, and topological
entropy.

Our idea was (and still is) that coverability is a local rule. Hence a
natural question, related to dynamical systems and tilings, is: does
this local rule enforce some global order? (For a broader study of
this question in the general context of tilings, see
e.g.~\cite{DurandLevinShen2005MathInt}). Independence results are
negative answers to this question: coverability, which is a local
rule, does not imply global properties. This is why we focus on a
stronger notion, multi-scale coverability, in the last part of this
paper.

Our preliminary results (from~\cite{GamardRichomme2015Lata}) are
summarized in the following table. Here, $\bot$ means ``independent'',
$?$ means ``not treated yet'' and $\implies$ means ``implies something
about''.

\begin{center}
  \begin{tabular}{| c | c c c c |}
    \hline %
    & Aperiodicity & Uniform recurrence & Frequencies & Entropy
    \\
    \hline
    Coverability
    & $\bot$ & $?$ & $?$ & $?$
    \\
    Multi-scale
    & $\bot$ & $\implies$ & $?$ & $\implies$
    \\
    \hline
  \end{tabular}
\end{center}

In this article, we extend these results in various ways. We complete
our independence (and dependence) results with coverability. We also
show that multi-scale coverability implies the existence of uniform
frequencies. As a summary, we have get the following table.

\begin{center}
  \begin{tabular}{| c | c c c c |}
    \hline %
    & Aperiodicity & Uniform recurrence & Frequencies & Entropy
    \\
    \hline
    Coverability
    & $\bot$ & $\bot$ & $\bot$ & $\implies$
    \\
    Multi-scale
    & $\bot$ & $\implies$ & $\implies$ & $\implies$
    \\
    \hline
  \end{tabular}
\end{center}

\noindent
The paper is structured as follows.

First, we recall notation and definitions of classical properties of
symmetry on pictures, notably uniform recurrence, uniform frequencies
and topological entropy. Then we adapt an elementary proof from the
one-dimensional case to show that, in two dimensions, coverability is
independent from these properties. This proof relies on a very
specific cover; we conclude this section by showing that, for many
other covers, coverability implies zero topological entropy. This new
result might seem surprising, as it is different from dimension one.
We conjecture that coverability may imply zero topological entropy,
except for a very specific class of covers (which are essentially
one-dimensional words). This would make things very different from the
one-dimensional case. %
(Section~\ref{sec:coverability})

In the previous section, we observed that independence of coverability
from other properties depend on the cover. Therefore, we give
characterizations of all the covers for which the independence results
about aperiodicity, uniform recurrence and uniform frequencies hold.
(Our previous work only did this for aperiodicity). %
(Section~\ref{sec:construct})

Finally, we study relations between multi-scale coverability and
topological entropy, uniform recurrence and uniform frequencies.
Multi-scale coverability is a good notion of symmetry in one
dimension, as it implies uniform recurrence, uniform frequencies and
zero topological entropy. In our preliminary work, we have studied
links between multi-scale coverability and uniform recurrence and
topological entropy in one dimension. We present these results, along
with a new one: in two dimensions, multi-scale coverability also
implies the existence of frequencies. This proof is purely
combinatorics and does not involve ergodic theory. %
(Section~\ref{sec:msc})

\section{Independence and Dependence Results}
\label{sec:coverability}

\subsection{Definitions and Notation}

In this section, we give all the notation and the definitions we will
use afterwards. The four next paragraphs adapt notation from
combinatorics on one-dimensional words to the context of pictures. The
rest of the section recalls definitions of classical properties of
symmetry on words, such as periodicity and uniform recurrence, adapted
for pictures. The most important one is the definition of
\emph{coverability}, which is central in this article. We will
occasionally use classical notation and well-known results from
combinatorics on one-dimensional words; for these,
see~\cite{Lothaire1997}.

\bigskip\noindent%
Let $\Sigma$ be a finite alphabet. An \emph{infinite picture} (or
\emph{two-dimensional word}, or $\bigZ^2$-word) is a function from
$\bigZ^2$ to $\Sigma$. Unless otherwise stated, those functions are
assumed to be total. When clarification is needed, we note $\dom(\bw)$
the domain of $\bw$, i.e. the set of coordinates where it has defined
letters.

A \emph{finite picture} (or \emph{block}, or \emph{rectangular word})
is a function $w$ such that
$\dom(w) = \{i, \dots, i + n - 1\} \times \{j, \dots, j + m - 1\}$,
for $i, j \in \bigZ$ and $n, m \in \bigN$. In that case, let
$\width(w) = n$ and $\height(w)=m$. Moreover, the function $w$ such
that $\dom(w) = \emptyset$ is \emph{the empty block}; it is considered
as a block and has width and height equal to $0$. Conversely, any
block which has either width or height equal to $0$ is the empty
block. The set of blocks of dimension $n \times m$ is denoted by
$\Sigma^{n \times m}$. More generally, if $u$ is a block, then
$u^{n \times m}$ denotes the $n\width(u) \times m\height(u)$-block
which consists only in occurrences of $u$.

If $u$ and $v$ are blocks, then let $|u|_v$ denote the number of
occurrences of $v$ in $u$. Let $u[x,y]$ denote the image of $(x,y)$ by
$u$. If $\bw$ is an infinite picture,
$\bw[(x,y), \cdots, (x+w-1, y+h-1)]$ denote the restriction of $\bw$
to the rectangle $\{x, \dots, x+w-1\} \times \{y, \dots, y+h-1\}$, for
$x,y \in \bigZ$ and $w,h \in \bigN$. If either $w$ or $h$ equal $0$,
then this denotes the empty block.

We will sometimes need to see finite blocks as one-dimensional words
whose alphabets are ``columns'' or ``lines''. %
Let $\bigC_{\Sigma, n}$ (resp. $\bigL_{\Sigma,m}$) denote the set of
$n$-columns (resp. $m$-lines) over $\Sigma$, i.e. $1 \times n$-blocks
(resp. $m \times 1$-blocks) over $\Sigma$. Concatenation in $\bigC$
and $\bigL$ is done respectively horizontally or vertically.

\bigskip\noindent %
In what follows, let $\bw$ be an infinite picture and let $u,v$ be
blocks. We recall some classical notions from combinatorics on words,
adapted to the two-dimensional case.

By definition, $\bw$ has a \emph{vector of periodicity}
$(k,\ell) \in \bigZ^{2} \backslash \{(0,0)\}$ if, for all positions
$(x,y) \in \bigZ^2$, we have $\bw(x,y) = \bw(x+k, y+\ell)$. Moreover,
we say that $\bw$ is \emph{periodic} if it has at least \emph{two
  non-colinear} vectors of periodicity.

We say that $u$ is a \emph{cover} (or a \emph{covering pattern}) of
$\bw$ if, for all $(x,y) \in \bigZ^2$, there exists
$(i,j) \in \bigN^2$ with $0 \leq i < \width(u)$ and
$0 \leq j < \height(u)$ such that
$\bw[(x-i, y-j) \dots (x-i+\width(u)-1, y-j+\height(u)-1])$ is equal
to $u$. Intuitively, $u$ is a cover of $\bw$ when each position of
$\bw$ belongs to an occurrence of $u$. If $\bw$ has at least one
cover, then it is \emph{coverable}.

The picture $\bw$ is \emph{uniformly recurrent} if, for all
$k \in \bigN$, there exists some $\ell \in \bigN$ such that all
$k \times k$-blocks of $\bw$ appear in all $\ell \times \ell$-blocks
of $\bw$. Intuitively, this means that any block of $\bw$ appears
infinitely often with bounded gaps.

Let $c_\bw(n,m)$ be the number of different $n \times m$-blocks of
$\bw$. Note that $c_\bw$ is known as the \emph{block complexity
  function} of $\bw$, and links between periodicity and block
complexity are currently investigated (see
e.g.~\cite{Cassaigne1999DLT}). In this paper, we will focus on the
\emph{topological entropy} of $\bw$, which is the following quantity:
\begin{equation*}
  H(\bw) = \lim_{n \to \infty} \frac{\log_{|\Sigma|} c_\bw(n,n)}{n^2}
\end{equation*}
Intuitively, if $c_\bw(n,n) \simeq |\Sigma|^{\varepsilon n^2}$, then
$H(\bw) \simeq \varepsilon$. In other words, when the complexity
function of $\bw$ is polynomial, $\bw$ has zero entropy. This is a
classical regularity property on words, often used in the context of
dynamical systems. Note that there are several kinds of entropy, in
addition to topological entropy. However, topological entropy is the
canonical notion when studying finite-type sub-shifts, our area of
interest.

Finally, the \emph{frequency} of $u$ in $\bw$ is the following
quantity:
\begin{equation*}
  f_u(\bw) = \lim_{n \to \infty} \frac{|\bw[(-n,-n) \dots (+n, +n)]|_u}{n^2}
\end{equation*}
if it exists. If $f_u(\bw)$ exists for all blocks $u$ of $\bw$, then
$u$ is said to \emph{have frequencies} (or \emph{have uniform
  frequencies}). This is another common regularity property coming
from dynamical systems, where it is more often called \emph{unique
  ergodicity}.

\subsection{Coverability is Independent from Classical Notions of
  Symmetry}

Now let us warm up with an easy independence result, already known in
one-dimension (see~\cite{MarcusMonteil2006Arxiv}).

First, recall that $\Sigma$ is a finite alphabet and let $h$ denote a
function from $\Sigma$ to $\Sigma^{n \times m}$, for $n,m \in \bigN$.
Then, if $\bw$ is a picture (either finite or infinite), then $h(\bw)$
is the word $\bw'$ such that
$\bw'[(xn,ym), \cdots, (xn + n-1, ym + m-1)] = h(\bw[x,y])$ for all
$x,y \in \bigN$. We say that $h$ is a \emph{morphism on pictures}. We
will use such a morphism in the next proof.

\begin{proposition}
  \label{proposition:coverable-nothing}
  On infinite pictures, coverability is independent from uniform
  recurrence, existence of frequencies and topological entropy.
\end{proposition}

\begin{proof}
  For uniform recurrence, observe that $q=
  \begin{smallmatrix}
    b&b&a\\ b&b&b\\ a&b&b\\
  \end{smallmatrix}$
  is a cover of the non-uniformly recurrent word displayed on
  Figure~\ref{fig:non-ur}. With the same value of $q$, the
  $q$-periodic infinite picture is uniformly recurrent.
  
  Let $\bw$ be a two-dimensional word over $\{a,b\}$ with polynomial
  (resp. exponential, resp. double-exponential) complexity. Consider
  the following function:
  \begin{align*}
    \nu&(a) = ababaaba \\
    \nu&(b) = abaababa
  \end{align*}
  The image $\nu(\bw)$ has polynomial with the same degree (resp.
  exponential, resp. double-exponential) complexity and is
  $aba$-coverable (viewing $aba$ as a $3 \times 1$-block). Therefore,
  we can get either zero or positive topological entropy for coverable
  words.

  Finally, the word $\nu(a^{\bigZ^2})$ has frequencies for all its
  blocks. By contrast, if $\bw$ is a word having no frequencies for
  any block, then $\nu(\bw)$ has no frequencies either.
\end{proof}

\begin{figure}[ht]
  \centering
  \includegraphics[width=0.9\textwidth]{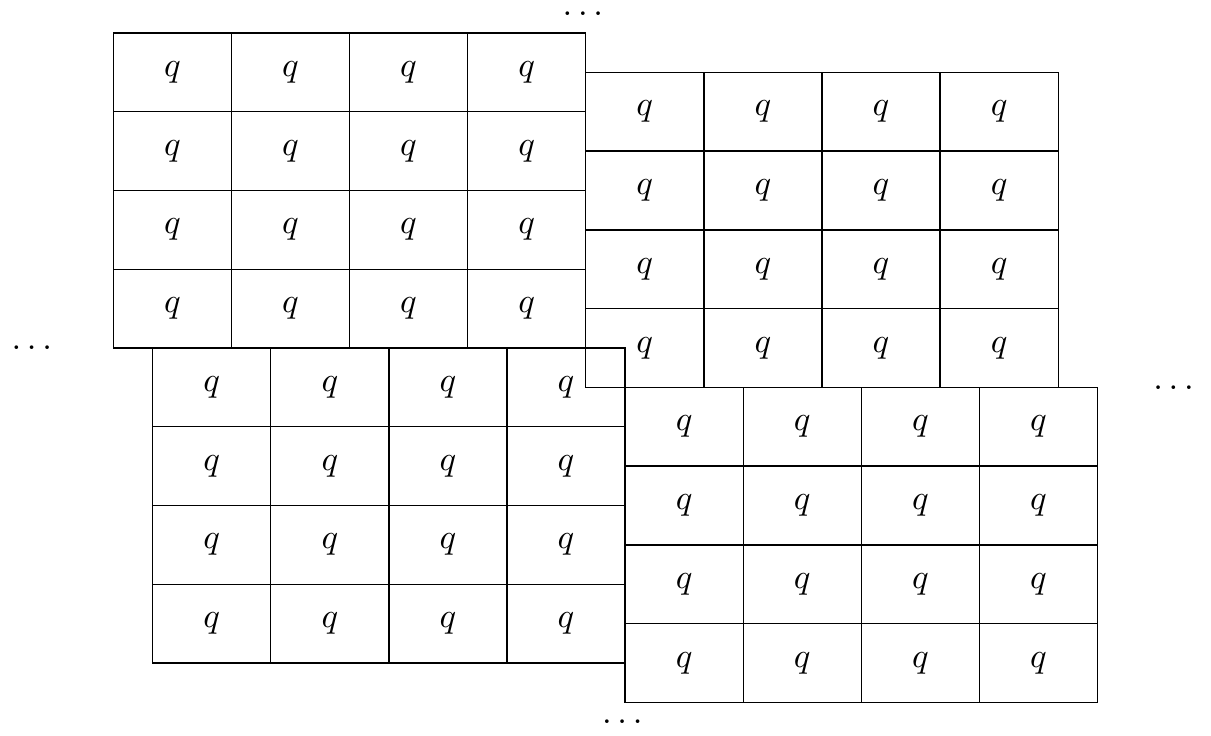}
  \caption{A coverable, non-uniformly recurrent word}
  \label{fig:non-ur}
\end{figure}

Proposition~\ref{proposition:coverable-nothing} shown independence
between coverability and various notions. However, the proof involved
a very specific quasiperiod, which was an $8 \times 1$-rectangle;
i.e., we reproduced the behaviour of one-dimensional coverability of
each line of an infinite picture. In next section, we use more
specifically the two dimensions and show that a reasonable condition
on the cover $q$ can enforce a global property of $q$-coverable
words.

\subsection{Topological Entropy of Coverable Pictures}
\label{sec:qp-ent0}

Let $q$ and $u$ be a finite blocks such that $q \neq u$ and $q$ is not
empty. We say that $u$ is \emph{a border of} $q$ when $u$ occurs in
two opposite corners of $q$. We are going to show that, if $q$ has a
corner without any (non-empty) border, then all $q$-coverable pictures
have zero topological entropy.

This is not a contradiction with
Proposition~\ref{proposition:coverable-nothing}, as we impose a
condition on the cover $q$. However, there are no equivalent results
in one dimension: no non-trivial condition on $q$ can force
topological entropy to be $0$ on right-infinite words. Therefore this
is a striking contrast with the one-dimensional case: even though
coverability and topological entropy are independent, the latter is
forced to be zero for a large class of covers. Hence global order
might arise from coverability if the cover is sufficiently
well-chosen.

\begin{theorem}
  \label{theorem:qp-ent0}
  Let $q$ be a finite picture of size $w \times h$ and $\bw$ an
  infinite $q$-coverable picture. If $q$ has a corner without any
  non-empty border, then $\bw$ has zero topological entropy.
\end{theorem}

Suppose, without loss of generality, that all the borders are in the
top right-hand corner of $q$. Moreover, $q$ has no full-width nor
full-height borders.

In what follows, ``\emph{occurrence}'' denotes an occurrence of $q$ in
$\bw$~---~unless otherwise stated. Moreover, $\occ(x,y)$ denotes the
occurrence of $q$ which covers the letter at coordinates $(x,y)$. If
there are several such occurrences, we choose the leftmost one among
the lowest ones. The \emph{coordinates} of an occurrence (or a block)
are the coordinates of its bottom, left-hand corner.

Let $o$ be an occurrence whose domain is
$\{(x,y), \cdots, (x+w-1, y+h-1)\}$. We denote
$\righ(o) = \occ(x+w,y+h-1)$. If $\righ(o)$ has coordinates $(x',y')$,
then $d_R(o) = y' - y$. In particular, the coordinates of $\righ(o)$
are $(x+w, y+d_R(o))$ (since there are no borders in this corner, see
Figure~\ref{fig:illus-base}). Similarly, we denote
$\abov(o) = \occ(x+w-1,y+h)$. If $\abov(o)$ has coordinates
$(x'',y'')$, then $d_A(o) = x'' - x$. The coordinates of $\abov(o)$
are $(x+d_A(o), y+h)$ (for the same reasons).

Observe that if both $d_A(o) > 0$ and $d_R(o) > 0$, then
$(d_A(o), d_R(o))$ is the dimensions of a border of $q$. So, there are
only three possible cases, illustrated by Figure~\ref{fig:illus-base}:
\begin{enumerate}
\item $d_R(o) = 0$ and $0 \leq d_A(o) < w$;
\item $d_A(o) = 0$ and $0 \leq d_R(o) < h$; or
\item $(d_A(o), d_R(o))$ is the dimensions of a border.
\end{enumerate}

\begin{figure}[h]
  \centering
  \includegraphics[width=0.9\textwidth]{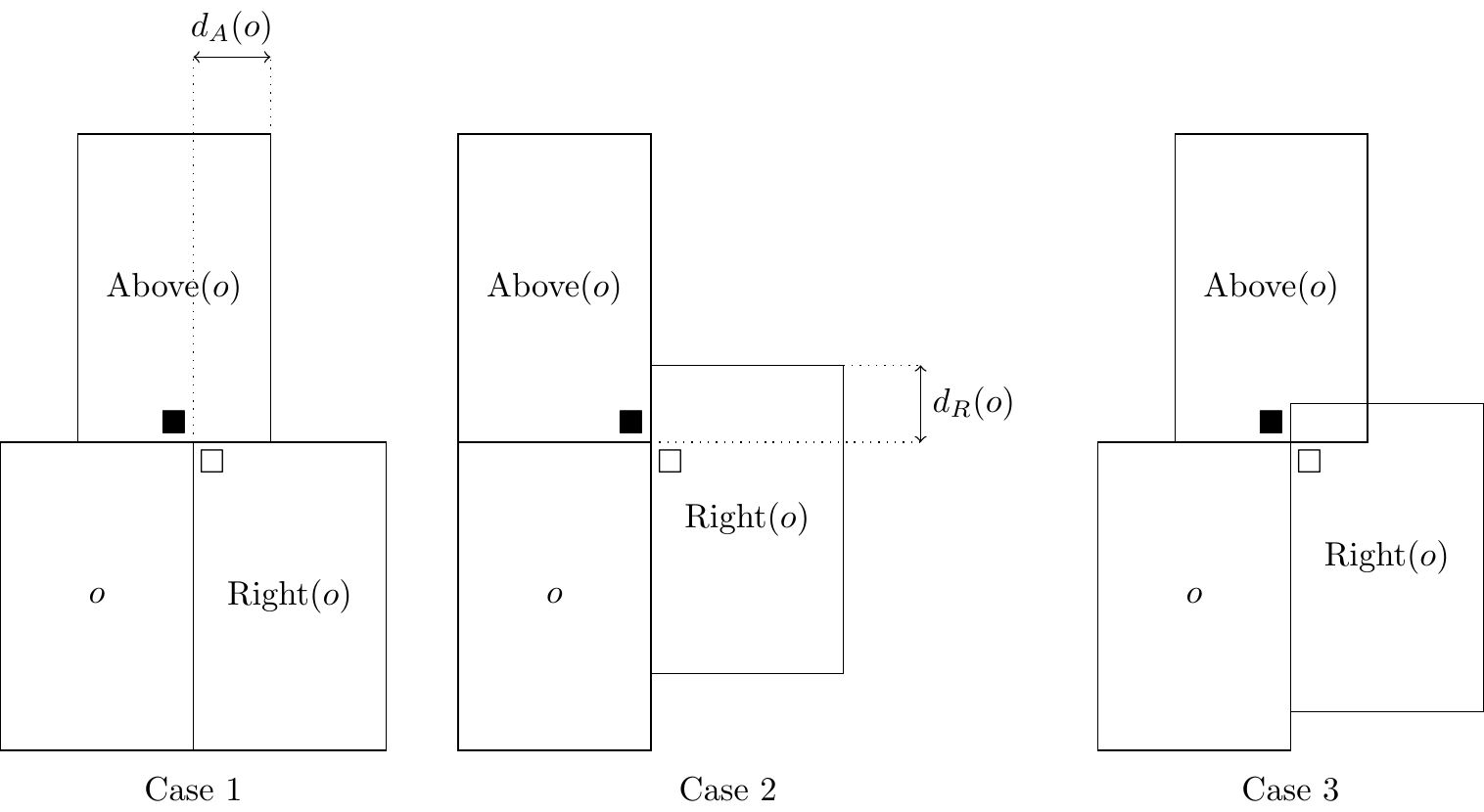}
  \caption[Illustration of basic definitions]{Illustration of $\abov$,
    $\righ$, $d_A$ and $d_R$. The squares $\square$ and $\blacksquare$
    show which letter $\righ(o)$ and $\abov(0)$ should cover,
    respectively.}
\label{fig:illus-base}
\end{figure}

\begin{proposition}
  \label{prop:uniq-def}
  Let $\bw$ and $q$ be as in Theorem~\ref{theorem:qp-ent0}.

  Let $o$ be an occurrence with coordinates $(x,y)$. Then we have:
  \begin{equation*}
    \righ(\abov(o)) = \abov(\righ(o)) = \occ(x+w+d_A(o), y+h+d_R(o))
  \end{equation*}
  In particular, this occurrence is uniquely determined by $o, d_A(o)$
  and $d_R(o)$.
\end{proposition}
\begin{proof}
  By definitions, we have:
  \begin{align*}
      & \abov(\righ(o)) \\
    = & \abov(\occ(x+w, y+d_R(o)) \\
    = & \occ(x + 2w - 1, y + h + d_R(o))
  \end{align*}
  and:
  \begin{align*}
      & \righ(\abov(o)) \\
    = & \righ(\occ(x+d_A(o), y+h)) \\
    = & \occ(x + w + d_A(o), y + 2h - 1)
  \end{align*}

  Consider $o' = \occ(x+w+d_A(o), y+h+d_R(o))$: its coordinates must
  be $(x+w+d_A(o), y+h+d_R(o))$, otherwise, it would overlap in an
  incorrect way (that is, on the upper right-hand corner) with either
  $\abov(o)$ or $\righ(o)$ (see Figure~\ref{fig:uniq-def}). So the top
  right-hand corner of $o'$ must be at coordinates
  $(x+2w+d_A(o)-1, w+2h+d_R(o)-1)$. We can see that both
  $(x + 2w - 1, y + d_R(o) + h)$ and $(x + 2 + d_A(o), y + 2h - 1)$
  are covered by $o'$ with some coordinates-checking:
  \begin{align*}
    x + w + d_A(o) &\leq x + 2w - 1 \leq x+2w+d_A(o)-1
    \\
    y + h + d_R(o) &\leq y + 2h - 1 \leq w+2h+d_R(o)-1
  \end{align*}
  Moreover, $o'$ is the lowest and the most leftwise occurrence
  covering those letters (any other one would overlap either
  $\righ(o)$ or $\abov(o)$ in an incorrect way, see
  Figure~\ref{fig:uniq-def} again). Hence, by definition of $\occ$,
  $o'$ must be the correct one.

\end{proof}

\begin{figure}[h]
  \centering
  \includegraphics[width=0.9\textwidth]{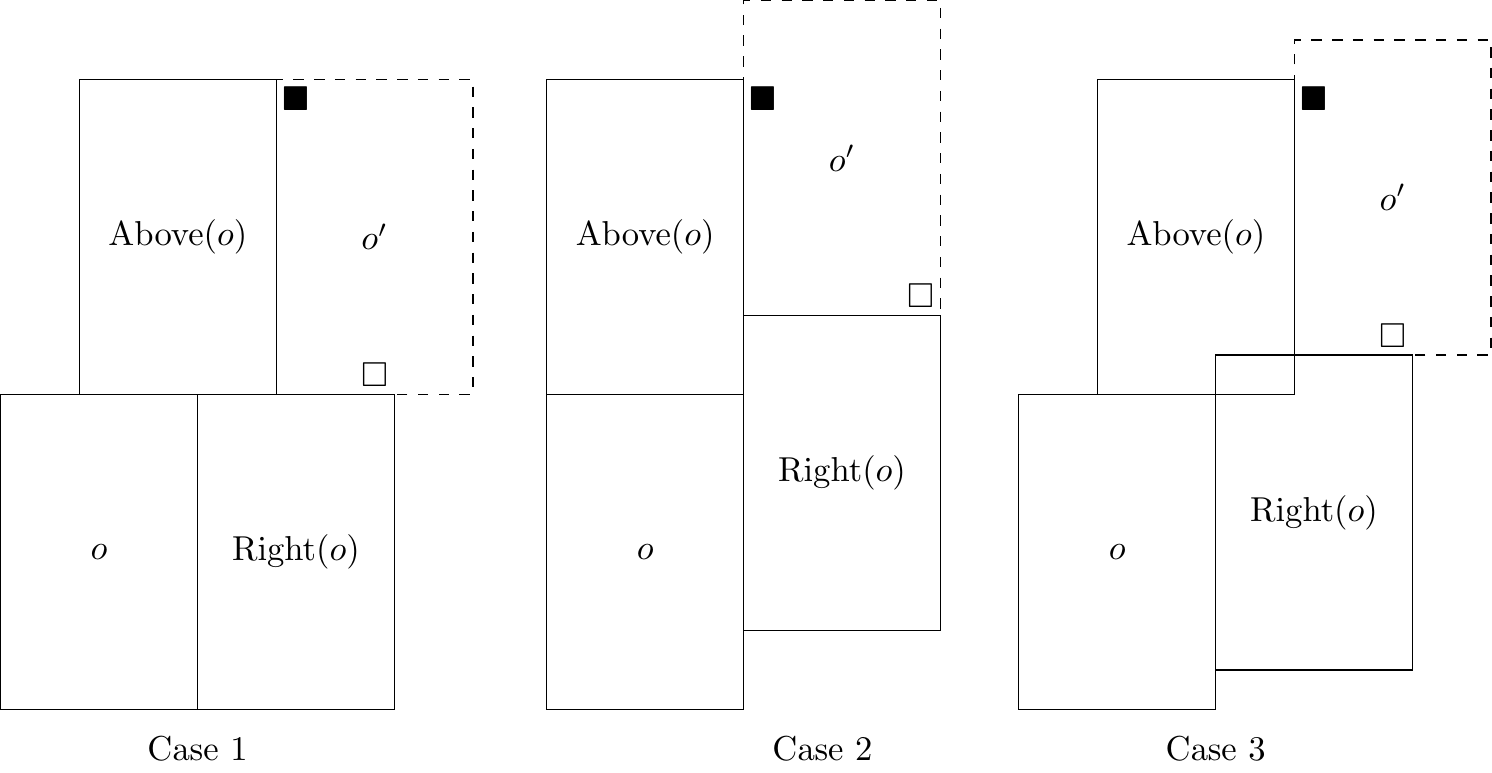}
  \caption[Illustration for proof of Proposition~\ref{prop:uniq-def}]{
    Illustration for proof of Proposition~\ref{prop:uniq-def}. The
    $\blacksquare$ indicates $(x+d_A(o)+w,y+2h-1)$ and $\square$
    indicates $(x+2w-1, y + d_R(o) + h)$.}
\label{fig:uniq-def}
\end{figure}

\begin{proof}[Proof of Theorem~\ref{theorem:qp-ent0}]
  Consider an arbitrary $n \times n$-block $B$ of $\bw$ whose leftmost
  bottom coordinate if $(x,y)$. Since $\bw$ is $q$-coverable, all
  letters of $B$ are covered by occurrences of $q$. Call
  $u_1, \cdots, u_k$ the occurrences of $q$ covering
  $B[(x,y), \cdots, (x+n-1,y)]$ (the bottom frontier of $B$), from
  left to right. Likewise, call $v_1, \cdots, v_\ell$ the occurrences
  of $q$ covering $B[(x,y), \cdots, (x,y+n-1)]$ (the left-hand
  frontier of $B$), from bottom to top. Observe that $u_1 = v_1$.

  By Proposition~\ref{prop:uniq-def}, $B$ is uniquely determined by
  $u_1, \dots, u_k$, $v_1, \dots, v_k$ and the coordinates of $B$
  relative to $u_1$. Let us bound the number of possible such $B$'s.
  
  For all $1 \leq i < k$, we have either $u_{i+1} = \righ(u_i)$ or
  $\abov(u_{i+1}) = \righ(u_i)$. Therefore, the number of possible
  sequences for $u_1, \dots, u_k$, given $u_1$, is bounded by
  $2^k \leq 2^n$. Likewise, for $1 \leq j < \ell$, we have
  $v_{j+1} = \abov(u_i)$ or $\righ(u_{i+1}) = \abov(u_i)$. Therefore,
  the number of possible sequences for $v_1, \dots, v_k$, given $v_1$,
  is bounded by $2^\ell \leq 2^n$. Finally, the number of possible $B$
  is bounded by $|q| \times 2^{2n}$.
  
  Observe that:
  \begin{equation*}
    \lim_{n \to \infty} \frac{\log (|q| \times 2^{2n})}{n^2} =
    \lim_{n \to \infty} \frac{2n \times \log |q|}{n^2}
    \to 0
  \end{equation*}
  therefore, $\bw$ has zero topological entropy.
\end{proof}

By contrast with Theorem~\ref{theorem:qp-ent0}, any block $q$ with
only full-width (or only full-height) borders has either $q$-coverable
pictures with positive entropy, or only periodic $q$-coverable
pictures. Proposition~\ref{proposition:coverable-nothing} is an
example of the first case. We conjecture that, in the remaining case
(if $q$ is a block without full-height or full-width borders, but with
borders in all corners), the topological entropy of the coverable
words is also zero.

Observe that, in the previous remark, we had to exclude the ``only
periodic coverable pictures'' case. This is a common problem with
coverability, which we will address in the next section.

\section{Aperiodic Coverings}
\label{sec:construct}

Here is a natural question about coverability: let $q$ be a block with
some property; do all $q$-coverable infinite pictures get some other
property? In other terms, can a property on a cover enforce another
property to all coverable pictures? For instance, we could imagine
that some property on $q$ would force each $q$-coverable picture to
have uniform frequencies.

Theorem~\ref{theorem:qp-ent0} is already a partial answer to this
question. In this section, we give a more exhaustive answer. We show
that there exist aperiodic $q$-coverable pictures if and and only if
the smallest block which tesselates $q$ does not overlap itself. This
is an answer to the previous question in the case of periodicity. This
is specifically important, as we generally want to exclude the
periodic case when working with coverability.

Finally, we observe that this condition is also necessary and
sufficient to force the existence of coverable non-uniformly recurrent
pictures, and even coverable non-uniform frequencies pictures.
Therefore, except for a trivial class of covers (trivial in the sense
that they only allow periodic coverable pictures), coverability does
not force any interesting property but (possibly) zero topological
entropy.

\subsection{A Condition for Simpler Cases}

First, let us consider the question on one-dimensional words (either
right-infinite or bi-infinite, it does not matter). Let $q$ be a
finite word. Is there a condition on $q$ forcing all $q$-coverable
one-dimensional words to be periodic?

Recall that, in the context of one-dimensional words, a \emph{border}
is a proper factor of $q$ which is both a prefix and a suffix of $q$.
(A word $u$ is a proper factor of $v$ if it is a subword of $v$ and
$u \neq v$). Moreover, the \emph{primitive root} $r$ of $q$ is the
shortest word such that, for some $k \in \bigN$, we have $q = r^k$. We
might have $q=r$ and $k=1$ for some words.

Now let us get some intuition about what is next. Let $q$ be a finite
(one-dimensional) word and $u$ be smallest border of $q$, so that we
have $q = uvu$ for some $v \in \Sigma^*$. Then, the word
$uvuv(uvu)^\omega$ is $q$-coverable (it is covered with occurrences of
$uvu$) seems not to be periodic ($uvuvuv$ only appears visually once).
However, this is not quite true: if there exists some
$r \in \Sigma^*, k,\ell \in \bigN$ such that $u = r^k$ and
$v = r^\ell$, then $uvuv(uvu)^\omega = r^\omega$, which is periodic.
This is why the condition is as follows.

\begin{proposition}
  \label{prop:N-case}
  Let $q$ be a finite one-dimensional word. There exists an aperiodic
  $q$-coverable $\bigN$-word if and only if the primitive root of $q$
  has a non-empty border.
\end{proposition}

\begin{proof}
  First, suppose that $q$ is a cover of an aperiodic infinite word
  $\bw$. Call $r$ the primitive root of $q$; observe that $r$ is also
  a cover of $\bw$. Suppose by contradiction that $r$ does not have
  any non-empty borders; then, two occurrences of $r$ never overlap.
  Hence $\bw$, which is $r$-covered, only consists in concatenations
  of $r$. Therefore, $\bw$ is $r$-periodic: a contradiction.
  
  Conversely, write $q = r^k$ with $r$ primitive and $k \geq 1$.
  Suppose that $r$ has a non-empty border and let $u$ be the smallest
  one, i.e. $r=uvu$ for some non-empty word $v$. Let $h$ be the
  morphism defined by $h(a) = (uvu)^k$ and $h(b)= u(vu)^{k}$. Both
  $h(a)$ and $uvu \cdot h(b)$ are $q$-coverable, so the image of any
  word begning with $a$ by $h$ is $q$-coverable. Moreover, since
  $r=uvu$ is a primitive word, $u(vu)^k \neq (uvu)^\ell$ for all
  $k, \ell \in \bigN$. Therefore, $h$ is injective, so the image of
  any aperiodic word by $h$ is also aperiodic.
\end{proof}

Now, let us shift to $\bigZ^2$-words. This shift is mainly motivated
by the study of tilings, which is why we chose $\bigZ^2$-words instead
of $\bigN^2$-words. We need some definitions before getting to proof
of our main theorem.

\subsection{Preliminaries for the $\bigZ^2$-Case}

Let $q$ and $r$ be blocks. In this context, $r$ is a \emph{root} of
$q$ if $q = r^{n \times m}$, for some positive integers $n$ and $m$.
If $q$ has no roots except itself, it is said to be \emph{primitive}.
These notions initially came from combinatorics on one-dimensional
words. The following lemma is a classical result about roots in one
dimension: it shows that any one-dimensional finite word has a
smallest root, called its \emph{primitive root}.

\begin{lemma} \label{lemma:lothaire}
  \emph{(See, e.g.,~\cite{Lothaire1997}, Prop.~1.3.1 and~1.3.2.)}
  \\ % 
  Given any finite one-dimensional words $u$ and $v$, the following
  statements are equivalent:
  \begin{enumerate}
  \item there exist integers $n,m \geq 0$ with $(n,m) \neq (0,0)$,
    such that $u^n = v^m$;
  \item there exist a word $t$ and positive integers $k$ and $\ell$
    such that $u = t^k$ and $v = t^\ell$;
  \item $uv=vu$.
  \end{enumerate}
\end{lemma}

\noindent
Let us show that primitive roots are also well-defined on rectangular
words.

\begin{lemma}
  \label{lemma:primitive-root}
  Let $q$ be a rectangular word. Suppose that $q$ has two distinct
  roots $r_1$ and $r_2$. Then there exists a rectangular word $r$ such
  that $r$ is a root of both $r_1$ and $r_2$.
\end{lemma}
\begin{proof}
  Let $r_1^k$ (resp. $r_2^k$) denote $k$ occurrences of $r_1$ (resp.
  $r_2$) concatenated vertically. Since $r_1$ and $r_2$ are roots of
  $q$, there exist integers $n$ and $m$ such that both $r_1^n$ and
  $r_2^m$ are roots of $q$, with $\height(q) = \height(r_1^n) =
  \height(r_2^m)$. Consider $q$, $r_1^n$ and $r_2^m$ as words over
  $\bigC_{\Sigma,\height(q)}$; by Lemma~\ref{lemma:lothaire}, there
  exists a word $c$ over $\bigC_{\Sigma,\height(q)}$ such that $c$ is
  a root of both $r_1^n$ and $r_2^m$.
  
  Let $r_3$ (resp. $r_4$) be the horizontal prefix of $r_1$ (resp.
  $r_2$) of length $\width(c)$. Both $r_3$ and $r_4$ are prefixes of
  $q$, hence $r_3^n = r_4^m$ (the power is still taken for vertical
  concatenation). Now view $r_3$ and $r_4$ as words over
  $\bigL_{\Sigma,\width(c)}$. By Lemma~\ref{lemma:lothaire}, there
  exists a word $r$ over $\bigL_{\Sigma,\width(c)}$ which is a common
  root of $r_3$ and $r_4$.
  
  As $r_1$ (resp. $r_2$) is obtained by horizontal concatenations of
  occurrences of $r_3$ (resp. $r_4$), we deduce that $r$ is a root of
  $r_1$ and of $r_2$.
\end{proof}

The \emph{primitive root} of a block $q$ is the root minimal for the
``is a root of'' relation. By Lemma~\ref{lemma:primitive-root}, it is
the only root of $q$ which is primitive. Note that $q$ might be its
own primitive root.

Finally, we say that a proper block $b$ of $q$ is a \emph{border} of
$q$ if $b$ occurs in two opposite corners of $q$. Note that it is
possible to have either $\width(b)=\width(q)$ (which we call a
\emph{horizontal border}) or $\height(b)=\height(q)$ (a \emph{vertical
  border}), but not both. If neither case applies, we call $b$ a
\emph{diagonal border}, following the terminology
from~\cite{CrochemoreIliopoulosKorda1998Algorithmica}.

\subsection{Blocks Covering Aperiodic Infinite Pictures}

Now we can state the condition under which a rectangular word can be
the covering pattern of a non-periodic $\bigZ^2$-word.

\begin{theorem}
  \label{theorem:construct}
  Let $q$ be a finite, rectangular word. Then there exists a
  $q$-co\-ve\-ra\-ble, non-periodic $\bigZ^2$-word if and only if the
  primitive root of $q$ has a non-empty diagonal border.
\end{theorem}

\noindent
This subsection is entirely dedicated to the proof of
Theorem~\ref{theorem:construct}.

\paragraph{Proof of the ``only if'' part}

First, suppose that $\bw$ is a $\bigZ^2$-word which is both
$q$-coverable and non-periodic. There exists at least two overlapping
occurrences of $q$ in $\bw$ (otherwise, $\bw$ would be $q$-periodic).
Moreover, the overlapping part is not a power of the primitive root of
$q$: if all overlappings were powers of some root $r$ of $q$, then
$\bw$ would be $r$-periodic. Therefore, $q$ must have at least one
border which is not a power of its primitive root. Hence its primitive
root has a non-empty border.

\paragraph{Proof of the ``if'' part}

Suppose that $q$'s primitive root has a non-empty diagonal border. Let
us build an infinite $\bigZ^2$-word which is $q$-coverable, but not
periodic.

\begin{figure}[ht]
  \centering
  \includegraphics[width=0.9\textwidth]{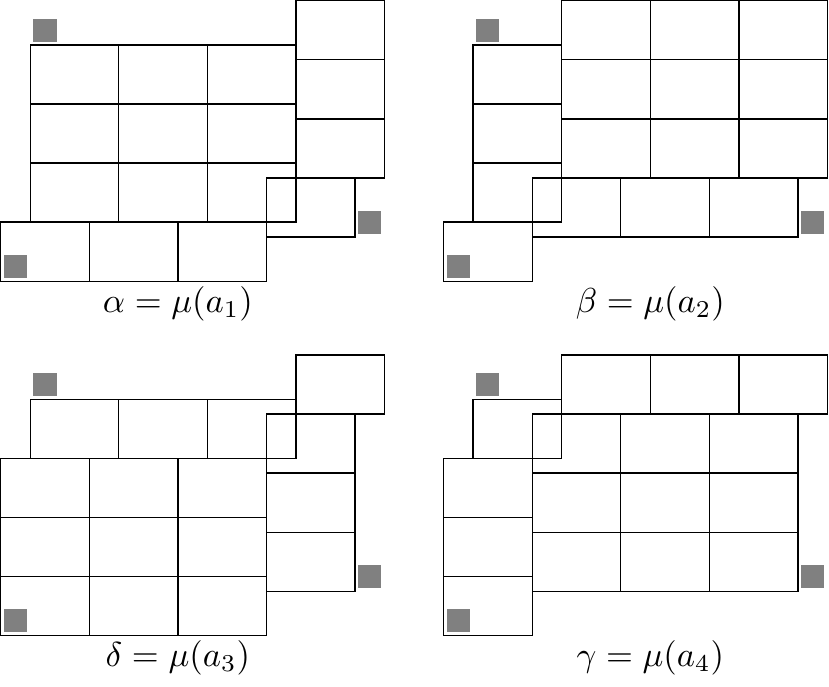}
  \caption{Four tiles to build a $q$-coverable word. Each rectangle is
    an occurrence of $q$.}
  \label{fig:gadgets}
\end{figure}

Let $r$ be the primitive root of $q$ and $b$ be a non-empty diagonal
border of $r$. Consider the four tiles $\alpha$, $\beta$, $\delta$ and
$\gamma$ displayed on Figure~\ref{fig:gadgets}. Each rectangle is an
occurrence of $q$. The overlapping zones are all occurrences of $b$
and the shifts on tile borders are sized accordingly. If the border
$b$ is on the opposite corner, all tiles are built symmetrically.

Let $A = \{a_1,a_2,a_3,a_4\}$ and $\mu$ be the function from
$A^{\bigZ^2}$ to $\Sigma^{\bigZ^2}$, defined by $\mu(a_1)=\alpha$,
$\mu(a_2)=\beta$, $\mu(a_3)=\delta$ and $\mu(a_4)=\gamma$. If its
input is regular enough, $\mu$ behaves more or less like a morphism,
with the following concatenation rules.

On Figure~\ref{fig:gadgets}, each tile has three \emph{anchors}, i.e.
letters marked by a small square. Concatenate two tiles horizontally
by merging the right-anchor of the first one with the left-anchor of
the second one. Concatenate two tiles vertically by merging the
bottom-anchor of the first one with the top-anchor of the second one.

More formally, we have:
\begin{align*}
  \mu(a_i \cdot u) &=
  \mu(a_i) \cup S_{(4\width(q); \height(b))} \circ \mu(u) \\
  \mu \begin{pmatrix}
    u \\ v
  \end{pmatrix} & = \mu(v) \cup S_{(\width(b);4\height(q))}
  \circ \mu(u)
\end{align*}
where $S_{(x,y)}$ denotes the translation (shift) by the vector
$(x,y)$ and the operator $\cup$ denotes the superposition of two
finite words. Recall that we view two-dimensional words as (possibly
partial) functions from $\bigZ^2$ to the alphabet. These functions
have domains which may be strictly included in $\bigZ^2$. If $w_1$ and
$w_2$ are two words with disjoints domains, then
$(w_1 \cup w_2)[x,y] = w_1[x,y]$ where $w_1$ is defined and $w_2[x,y]$
where $w_2$ is defined. In what follows, we will only consider
superpositions where no position $(x,y)$ is defined in both $w_1[x,y]$
and $w_2[x,y]$.

If $u$ is a rectangular word, the leftmost bottom anchor of
$\mu(u[i,j])$ has coordinates:
\begin{equation*}
 (i \times 4 \times \width(q) + j \times \width(b);\;
   j \times 4 \times \height(q) + i \times \height(b))
\end{equation*}
in $\mu(u)$. Figure~\ref{fig:gadgets-exe} gives an example of how
$\mu$ works.

\begin{figure}[ht]
  \centering
  \includegraphics[width=0.9\textwidth]{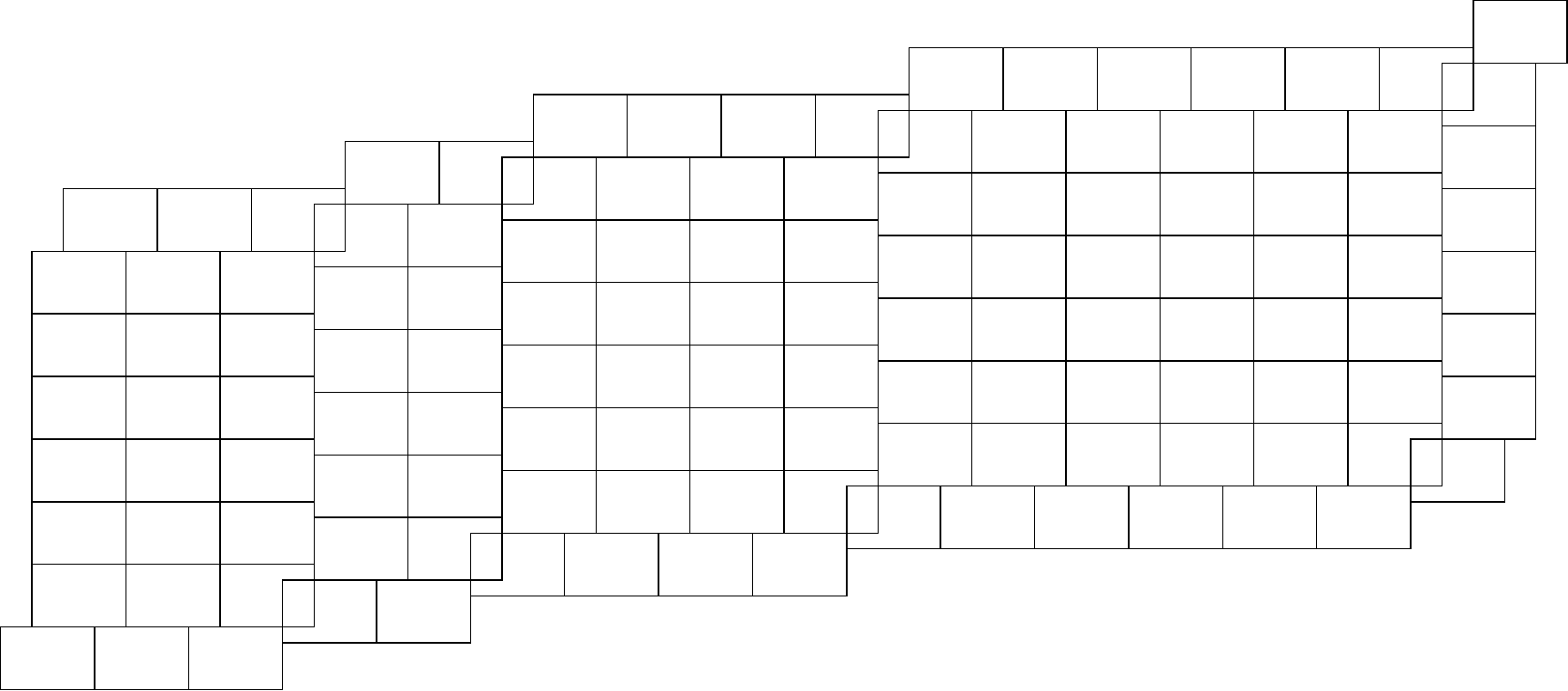}
  \caption[Image of a simple word by $\mu$.]{$\mu\left(\begin{smallmatrix}
      a_3&a_4&a_4&a_3\\a_1&a_2&a_2&a_1\\
    \end{smallmatrix}\right)$, each rectangle is an occurrence of $q$}
  \label{fig:gadgets-exe}
\end{figure}

An infinite picture over $A$ is \emph{suitable} when it satisfies the
following conditions:
\begin{enumerate}
\item each line is either on alphabet $\{a_1, a_2\}$ or on alphabet
  $\{a_3, a_4\}$;
\item each column is either on alphabet $\{a_1, a_3\}$ or on alphabet
  $\{a_2, a_4\}$.
\end{enumerate}

Let us check that if $\bw$ is suitable, then each letter of $\mu(\bw)$
belongs to the image of exactly one letter of $\bw$. This essentially
means that all tiles ``fit together'' with no overlaps.

By construction, tiles $\alpha$ and $\delta$ fit together vertically,
and tiles $\beta$ and $\gamma$ fit as well. Hence
$\mu(\begin{smallmatrix} a_1 \\ a_3 \\ \end{smallmatrix})$ and
$\mu(\begin{smallmatrix} a_2 \\ a_4 \\ \end{smallmatrix})$ are
well-defined. Likewise, tiles $\alpha$ and $\beta$ fit together
horizontally, and tiles $\delta$ and $\gamma$ fit as well. Hence
$\mu(a_1 a_2)$ and $\mu(a_3 a_4)$ and are well-defined. Iterating this
argument, we deduce that the image of any suitable word is
well-defined.
  
Moreover, we let readers check that $\mu(\bw)$ has no ``holes''. More
precisely, if if $\bw$ is a suitable rectangular word, $\mu(w)$
satisfies the following weak convexity properties:
\begin{itemize}
\item for all $i,j,j_1,j_2 \in \bigN$ with $j_1 \leq j \leq j_2$, if
  $(i,j_1)$ and $(i,j_2)$ are in $\dom(\mu(\bw))$, then $(i,j)$ is in
  $\dom(\mu(\bw))$ as well;
\item for all $i,j,i_1,i_2 \in \bigN$ with $i_1 \leq i \leq i_2$, if
  $(i_1,j)$ and $(i_2,j)$ are in $\dom(\mu(\bw))$, then $(i,j)$ is in
  $\dom(\mu(\bw))$ as well.
\end{itemize}
  
As a consequence, the definition of $\mu$ can be extended to suitable
$\bigZ^2$-words. If $\bw$ is a suitable $\bigZ^2$-word, then
$\mu(\bw)$ is a well-defined $\bigZ^2$-word as well.

Now let us prove that if $\bw$, an infinite picture, is aperiodic,
then so is $\mu(\bw)$. First, we need a technical lemma about our
tiles.

\begin{lemma}
  \label{lemma:cannot-overlap}
  Let $x$ and $y$ be different tiles from
  $\{\alpha,\beta,\gamma,\delta\}$. Then an occurrence of $x$ and an
  occurrence of $y$ cannot overlap when their anchor points coincide.
\end{lemma}

\noindent
This essentially means that situations from
Figure~\ref{fig:overlap-others} cannot occur.

\begin{figure}[ht]
  \centering
  \includegraphics[width=0.8\textwidth]{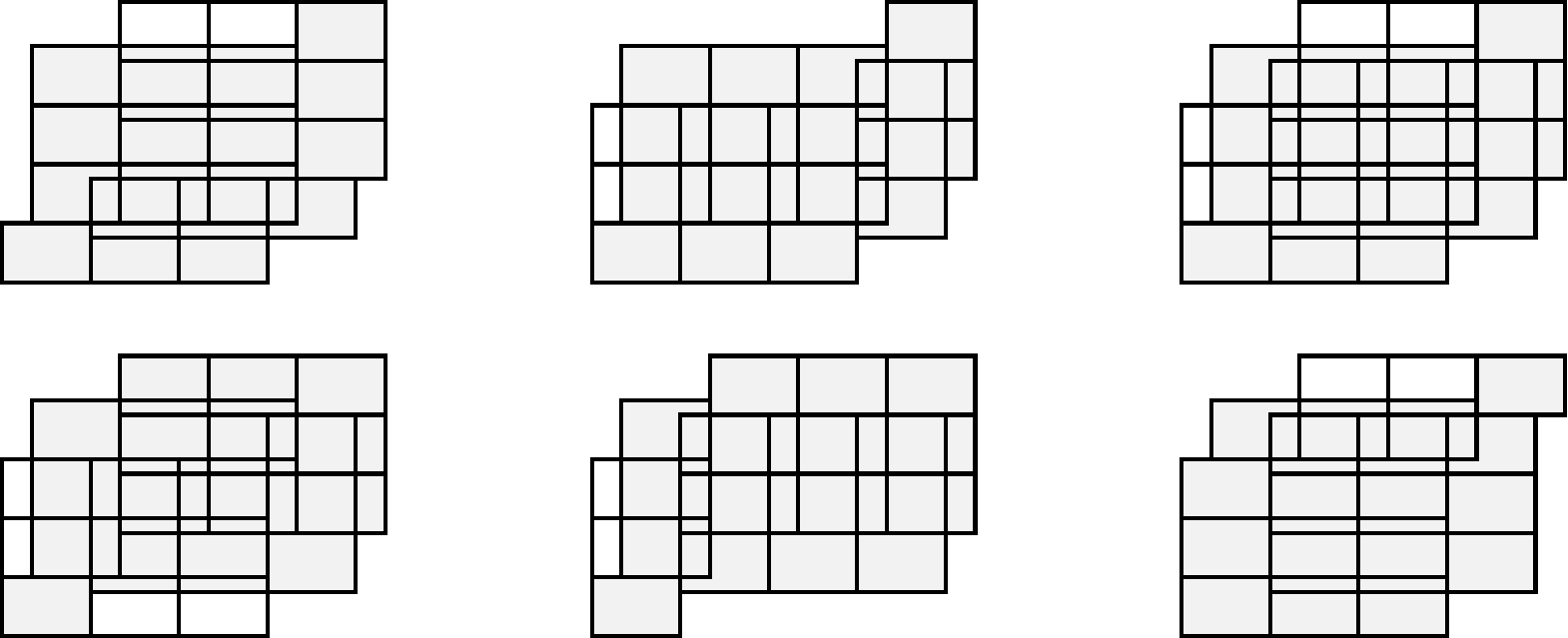}
  \caption{All other possible overlappings.}
  \label{fig:overlap-others}
\end{figure}

\begin{proof}
  There are six possibilities for the set $\{x,y\}$. All proofs are
  the same, up to some symmetry, so we only provide a proof when
  $x=\alpha$ and $y=\beta$ (illustrated by the top left-hand case of
  Figure~\ref{fig:overlap-others}). In what follows, $q$ refers to the
  rectangular word used for the construction of the tiles, $r$ to its
  primitive root and $b$ to a diagonal border of $r$.

  There are three occurrences of $q$, named $q_1$, $q_2$ and $q_3$,
  such that $q_1$ is covered by $q_2$ and $q_3$ and all three are
  horizontally aligned. In other words, $q_1$ is an occurrence of $q$
  in $q_2 \cdot q_3$ viewed as words over $\bigL_{\Sigma, \width(q)}$.
  (See for instance the top second column of $q$'s in the figure).
  View $q_1, q_2$ and $q_3$ as one-dimensional words over the alphabet
  $\bigL_{\Sigma,\width(q)}$. There exist words $x$ and $x'$ over
  $\bigL_{\Sigma,\width(q)}$ such that $q_1 = xx'$ and
  $q_2 = q_3 = x'x$ (where words are concatenated from bottom to top).

  By Lemma~\ref{lemma:lothaire}, $x$ and $x'$ (and $q$) are powers of
  a same word $s$ over $\bigL_{\Sigma,\width(q)}$. Notice that
  $\height(x') = \height(b)$ and $\height(x) = \height(q) -
  \height(b)$. It follows that $\height(s)$ divides $\height(x)$ and
  $\height(q) - \height(b)$.

  Observe that $s$ is a vertical prefix of both $q$ and $x$. Thence
  one can find three occurrences of $s$, named $s_1$, $s_2$ and $s_3$,
  such that $s_1$ is covered by $s_2$ and $s_3$ and all three are
  vertically aligned. In other terms, $s_q$ is an occurrence of $q$ in
  $s_2 \cdot s_3$ viewed as a word over $\bigC_{\Sigma,\height(q)}$.
  (See for instance the second line of $q$'s in the figure).

  Now view $s$ as a one-dimensional word on the alphabet
  $\bigC_{\Sigma,\height(s)}$. There exist words $y, y'$ such that
  $s_1 = yy'$, $s_2 = s_3 = y'y$ and $\width(y')=\width(b)$. By
  Lemma~\ref{lemma:lothaire}, we deduce that there exists a word $t$
  over $\bigC_{\Sigma,\height(s)}$ such that $y$ and $y'$ (and $s$)
  are powers of $t$.
  
  Let $k \leq 1$ be the integer such that $q = s^k$ (for vertical
  concatenation) and let $\ell \geq 1$ be the integer such that
  $s = t^\ell$ (for horizontal concatenation). We have that
  $q = t^{\ell \times k}$. Therefore $t$ is a root of $q$ such that
  $\width(t) \leq width( y') = width( b )$ and
  $\height(t) = \height(s) \leq \height(b)$. Thus
  $\width(s) \times \height(s) \leq \width(b) \times \height(b)$ which
  is a contradiction with the definition of $b$. Indeed, recall that
  $b$ is a border (hence a proper block) of the primitive root of $q$,
  which is the smallest (in number of letters) root of $q$.
\end{proof}

In the proof of next lemma, Lemma~\ref{lemma:cannot-overlap} helps to
establish a correspondence between the letters of the $\bigZ^2$-word
$\mu(\bw)$ and the ``tiling'' consisting of occurrences of $\alpha,
\beta, \delta$ and $\gamma$. We need this correspondence to prove that
some $\mu(\bw)$ can always be made aperiodic.

\begin{lemma}
  Let $q$ be a rectangular word, $r$ its primitive root and $b$ one
  non-empty diagonal border of $r$. Let $\bw$ be an aperiodic,
  suitable $\bigZ^2$-word. Then $\mu(\bw)$ is an aperiodic,
  $q$-coverable $\bigZ^2$-word.
\end{lemma}
\begin{proof}
  By construction, $\mu(\bw)$ is $q$-coverable for all $\bw$. Suppose
  that $\mu(\bw)$ has a non-zero vector of periodicity
  $\vect p \in \bigZ^2$. Let us prove that, under this assumption,
  $\bw$ is periodic.

  Let $a \in \bigZ^2$ be the coordinates of the anchor point of some
  tile in $\mu(\bw)$. For any $i \in \bigZ$, let $t_i = a + i \times
  \vect p$. Since tiles have at most $16 \times \width(q) \times
  \height(q)$ letters, by pigeonhole principle, there are two pairs of
  coordinates $t_i$ and $t_j$ which have the same offset to the anchor
  points of their respective tiles (i.e. the tiles covering their
  respective positions). Hence the difference between these anchor
  points is a multiple of the vector of periodicity $\vect p$.

  Let $T_i$ (resp. $T_j$) be the tile covering position $t_i$ (resp.
  $t_j$). Since $T_i$ is the $(j - i) \times \vect p$-translation of
  $T_j$, they are both occurrences of a same tile. Moreover, the
  right-neighbours of $T_i$ and $T_j$ are both occurrences of a same
  tile, otherwise we would have a configuration forbidden by
  Lemma~\ref{lemma:cannot-overlap}. Likewise, the top-neighbour,
  bottom-neighbour and left-neighbour of $T_i$ and $T_j$ are also
  equal. By iterating this argument over the neighbours' neighbours,
  and so on, we conclude that the tiling itself is periodic. Hence,
  $\bw$ is periodic.
\end{proof}

This ends the proof of Theorem~\ref{theorem:construct}. From any
rectangular word $q$ with at least one non-empty diagonal border in
its primitive root, we can build $\mu(\bw)$ for any aperiodic,
suitable $\bigZ^2$-word $\bw$. The picture $\mu(\bw)$ will be
$q$-coverable and aperiodic.

\subsection{Uniform Recurrence, Uniform Frequencies and Coverability}

In this subsection, we extend Theorem~\ref{theorem:construct} (which
characterizes periodicity in terms of covers) to uniform recurrence
and uniform frequencies. To do so, we exploit
Theorem~\ref{theorem:construct} itself as well as the function $\mu$
from its proof.

Let $q$ be a block whose primitive root has a non-empty diagonal
border. Let $\mu$ be as in the proof of
Theorem~\ref{theorem:construct}. Recall that $\mu$ takes as an
argument a picture over the alphabet $\{a_1, a_2, a_3, a_4\}$.
Consider the word $\bw$ such that:
\begin{itemize}
\item $\bw[0,0] = a_4$;
\item $\bw[0,j] = a_2$ for $j \in \bigZ^*$;
\item $\bw[i,0] = a_3$ for $i \in \bigZ^*$;
\item $\bw[i,j] = a_1$ for $i \in \bigZ^*$ and $j \in \bigZ^*$.
\end{itemize}
Observe that $\bw$ is suitable, hence $\mu(\bw)$ exists. Moreover,
$\mu(\bw)$ is not uniformly recurrent: any block which contains
$\mu(a_4)$ only occurs once. Therefore $\mu(\bw)$ is a $q$-coverable
picture not uniformly recurrent.

Now consider $t$ a $\bigZ-$word (one-dimensional) over $\{1,2\}$, such
that no factor of $t$ has uniform frequencies in $t$. Then define
$\bw'$ an infinite picture over $\{a_1, a_2, a_3, a_4\}$ by
$\bw'[i,j] = a_{t[i]}$, for all $i,j \in \bigZ$. Observe that, in
$\mu(\bw')$, the frequency of any block containing exactly one
occurrence of $\mu(a_2)$ is the frequency of $a_2$ in $t$. Therefore,
those factors do not have uniform frequencies.

\begin{proposition}
  Let $q$ be a block. There exists a $q$-coverable infinite picture
  which is not uniformly recurrent (resp. has no uniform frequencies)
  if and only if the primitive root of $q$ has a non-empty diagonal
  border.
\end{proposition}
\begin{proof}
  If $q$ does not have a primitive root with a diagonal border, then
  by Theorem~\ref{theorem:construct} all $q$-coverable pictures are
  periodic and therefore uniformly recurrent (resp. have uniform
  frequencies). Otherwise, the constructions above gives $\mu(\bw)$
  (resp. $\mu(\bw')$) a non-uniformly recurrent, and not even
  recurrent (resp. without uniform frequencies) $q$-coverable word.
\end{proof}

This remark is a negative result answering our initial question. As
soon as the cover $q$ we choose is non-trivial (has at least one
non-periodic coverable word), we have $q$-coverable word without
frequencies and without uniform recurrence. By contrast, we
conjectured in Section~\ref{sec:qp-ent0} that, for almost any $q$, all
$q$-coverable words have zero topological entropy. To sum up,
coverability is independent from uniform recurrence and uniform
frequencies for all ``non-trivial'' covers, and it likely implies zero
topological entropy for all ``almost all'' covers. Not only
coverability implies very little symmetry properties on words, but the
covers themselves bear very little information.

As a consequence, we move on a stronger notion of coverability, based
on the same idea, but with (hopefully) better properties.

\section{Multi-Scale Coverability in Two Dimensions}
\label{sec:msc}

In~\cite{MarcusMonteil2006Arxiv}, Monteil and Marcus called
\emph{multi-scale quasiperiodicity} any $\bigN$-word having infinitely
many quasiperiods. In our context, we want to exclude cases where
coverability is obtained on groups of one-dimensional lines (or
columns) packed all over $\bigZ^2$. Hence we call a $\bigZ^2$-word
\textit{multi-scale coverable} if, for each $n \in \bigN$, it has a
$k \times \ell$-cover with both $k \geq n$ and $\ell \geq n$. This is
actually a generalization of one-dimensional multi-scale
quasiperiodicity (in $1$D, all quasiperiods are prefixes of the
multi-scale quasiperiodic word, so the quasiperiods must be longer and
longer).

In~\cite{MarcusMonteil2006Arxiv}, Monteil and Marcus prove that
multi-scale coverable right-infinite words have uniform frequencies,
zero topological entropy and are uniformly recurrent. In this section,
we generalize implications of topological entropy and uniform
frequencies to the two-dimensional case. Then we see that uniform
recurrence is a bit more subtle.

\subsection{Topological Entropy}
\label{sec:ms-entropy}

Let $\bw$ be a $\bigZ^2$-word. Recall that $c_\bw(n,m)$ is the number
different of blocks of size $n \times m$ which occur in $\bw$ and that
the topological entropy of $\bw$ is the following quantity:
\begin{equation}
  \label{eq:entropy}
  H(]bw) = \lim_{n \to +\infty} \frac{\log_{|\Sigma|} c_w(n,n)}{n^2}
\end{equation}

\begin{proposition}
  \label{proposition:msqp-ent0}
  Any multi-scale coverable, $\bigZ^2$-word $\bw$ has zero topological
  entropy.
\end{proposition}

\begin{proof}
  Consider a covering pattern $q$ of $\bw$ with size $n \times m$.
  Suppose without loss of generality that $n \leq m$. Let $s$ be a $m
  \times m$-square of $\bw$. The square $s$ is covered with
  occurrences of $q$ (which may spill out of $s$). The relative
  position of $s$ and of occurrences of $q$ completely defines $s$.

  We need at most $4m$ occurrences of $q$ to define a covering of $s$.
  Indeed, each occurrence of $q$ must have at least one of its corners
  in $s$. If some occurrence of $q$ has its bottom right-hand corner
  in $s$, then no other occurrence of $q$ may have their bottom
  right-hand corners on the same line of $s$. Otherwise, one of these
  occurrences would supersede the other one, which would be
  ``useless'' in the covering. Proceed the same way for the other
  corners and deduce that at most $4m$ occurrences of $q$ ($4$ per
  line) uniquely define $s$.
  
  Each of these occurrences is uniquely determined by its position of
  its corner on a line of $s$. There are at most $m$ possibilities for
  each. Therefore, there are at most $m^{4m}$ $q$-coverings which
  define at most $m^{4m}$ squares of size $m \times m$.

  This bound on $c_\bw(m,m)$ allows us to compute the entropy of
  $\bw$. Observe that:
  \begin{equation}
    \label{eq:entropy-0}
    \lim_{m \to \infty} \frac{ \log m^{4m} }{ m^2 } =
    \lim_{m \to \infty} \frac{ 4m \log m }{ m^2 } \to
    0
  \end{equation}

  Since there are infinitely many covering patterns of $\bw$ with
  growing sizes, there are infinitely many integers $m$ such that
  $c_\bw(m,m) \leq m^{4m}$. Hence equation~\eqref{eq:entropy-0} shows
  that then topological entropy of $\bw$ converges to zero.
\end{proof}

Note that since the Kolmogorov complexity is bounded by the
topological entropy (see~\cite{Brudno1982Trudy}), this result also
shows that the Kolmogorov complexity of multi-scale coverable words is
zero as well.

\subsection{Uniform Frequencies}

In this subsection, we prove the following theorem:
\begin{theorem} \label{theorem:2d-freqs} %
  Multi-scale coverable pictures have uniform frequencies.
\end{theorem}

This answers an open question from our preliminary work on the
subject~\cite{GamardRichomme2015Lata}, and we generalize a result from
inifnite words to infinite pictures. We believe that this proof is
easily adaptable to higher dimensions ($3$-dimensional words, etc.).
Moreover, the proof on words in~\cite{MarcusMonteil2006Arxiv} was
expressed in terms of sub-shifts and used ergodic theory. By contrast,
our proof uses purely combinatorial means, and hence is accessible to
readers unfamiliar with this theory. For a full expos{\'{e}} about
ergodic theory in the context of one-dimensional words,
see~\cite{CANT7}.

\medskip
\noindent
The authors thank Ana{\"{e}}l~Grandjean for his precious help with
this proof.

\medskip
First, let us recall some notation. If $u$ and $v$ are finite blocks,
then $|u| = \width(u) \times \height(u)$ and $|v|_u$ is the number of
occurrences of $u$ in $v$. We note $f_u(v)$ the \emph{frequency} of
$u$ in $v$, which is the following quantity:
\begin{equation*}
  f_u(v) = \frac{|v|_u}{|v|}
\end{equation*}
Moreover, $u[x,y]$ is the letter at position $(x,y)$ of $u$ and
$u[ (x,y) \dots (x+w-1,y+h-1) ]$ is the block of $u$ whose bottom
left-hand corner is at position $(x,y)$ and whose size is $(w,h)$. Let
$\bw$ be an infinite picture. Define $B_n(\bw)$ as
$\bw[(-n,-n), \cdots, (+n,+n)]$ and $f_u(\bw)$ as:
\begin{equation*}
  f_u(\bw) = \lim_{n \to +\infty} f_u(B_n)
\end{equation*}
if this quantity exists.

We say that $\bw$ \emph{has frequencies} (or has \emph{uniform
  frequencies}, or is \emph{uniquely ergodic}) if $f_u(\bw)$ exists
for each block $u$ of $\bw$. Our purpose is to show that any
multi-scale coverable picture has frequencies.

The proof is structured in several lemmas. Lemmas~\ref{lemma:2d-fact},
\ref{lemma:2d-exists-big-things} and~\ref{lemma:2dbis-upper-bound} are
technical, their sole purpose is to make calculations work.
Lemma~\ref{lemma:2d-cases} states that, if $\bw$ is a multi-scale
coverable picture without frequencies, then there exist either:
\begin{enumerate}
\item infinitely many covers with high frequencies and infinitely many
  blocks with low frequencies; or
\item infinitely many covers with low frequencies and infinitely many
  blocks with high frequencies.
\end{enumerate}
Lemma~\ref{lemma:2d-upper-class} states that, if a picture has
infinitely many covers with high frequencies and infinitely many
blocks with low frequencies, then it has infinitely many blocks with
\emph{even lower} frequencies. Similarly,
Lemma~\ref{lemma:2d-lower-class} states that if a picture has
infinitely many covers with low frequencies and infinitely many blocks
with high frequencies, then it has infinitely many blocks with
\emph{even higher} frequencies. Finally, the proof of
Theorem~\ref{theorem:2d-freqs} is as follows: by
Lemma~\ref{lemma:2d-cases}, any multi-scale picture without
frequencies is either in case~$1$ or in case~$2$. In case~$1$, apply
Lemma~\ref{lemma:2d-upper-class} many times, until getting blocks with
frequencies higher than $1$: a contradiction. In case~$2$, apply
Lemma~\ref{lemma:2d-lower-class} many times, until getting blocks with
frequencies lower than $0$: a contradiction again.

\medskip\noindent Let us start. \medskip\noindent

If $\bw$ is an infinite picture, let $L_{\geq K}(\bw)$ the set of
blocks of $\bw$ whose width is larger than $K$ and whose height is
also larger than $K$.

For the purposes of the proof, we shall need to extend definitions of
$|v|$ and $|v|_u$ to cases where $v$ is a finite union of blocks
(instead of a single block). In that case, $|v|$ is the number of
letters in $v$ and $|v|_u$ the number of complete occurrences of $u$
in $v$.

\begin{lemma}
  \label{lemma:2d-fact} %
  Let $v$ be a finite block which can be decomposed into $v_1$ and
  $v_2$, two unions of blocks such that $\dom v_1 \cap \dom v_2 = \emptyset$.
  Let $u$ be a block. Then,
  \begin{equation*}
    f_u(v_1) \times \frac{|v_1|}{|v|} \leq
    f_u(v_1) \times \frac{|v_1|}{|v|} +
    f_u(v_2) \times \frac{|v_2|}{|v|} \leq
    f_u(v) \leq
    f_u(v_1) \times \frac{|v_1|}{|v|} + \frac{|v_2|}{|v|}
  \end{equation*}
\end{lemma}

\begin{proof}
  This is equivalent to the following inequation:
  \begin{equation*}
    |v_1|_u \leq |v_1|_u  + |v_2|_u \leq |v|_u \leq |v_1|_u + |v_2|
  \end{equation*}
  The first two inequality follows from $\dom v_1 \subseteq \dom v$
  and $\dom v_2 \subseteq \dom v$. For the last inequality, consider
  that each occurrence of $u$ in $v$ has its top right-hand corner
  which is either in $v_1$ or in $v_2$. The number of occurrences with
  their corner in $v_1$ is precisely $|v_1|_u$. In the ``worst'' case,
  there is an occurrence of $u$ per letter of $v_2$.
\end{proof}

We shall also need to express the relation between width, height and
areas of bigger and bigger blocks.

\begin{lemma}
  \label{lemma:2d-exists-big-things}
  Let $x,y,z,t$ be real numbers such that $x>0, y>0$. Let
  $(u_n)_{n \in \bigN}$ be a sequence of finite pictures such that,
  for all $i$, $\width(u_i) > i$ and $\height(u_i) > i$. Then, there
  exists an in\-te\-ger $N$ such that, for all $n > N$:
  \begin{equation*}
    t < \frac{\width(u_n) \times \height(u_n)}{
      x\width(u_n) + y\height(u_n) + z}
  \end{equation*}
\end{lemma}
\begin{proof}
  Observe that, since $x > 0$ and $y > 0$, the following functions are
  non-decreasing:
  \begin{align*}
    w \mapsto \frac{wh}{xw + yh + z}
    && h \mapsto \frac{wh}{xw + yh + z} 
  \end{align*}
  Hence, as for all $i$, $\width(u_i) > i$ and $\height(u_i) > i$, we
  have for arbitrarily large $n$:
  \begin{equation*}
    \frac{n^2}{(x+y) \times n + z} <
    \frac{n \times \height(u_n)}{n + y \times \height(u_n) + z} <
    \frac{\width(u_n) \times \height(u_n)}{x \width(u_n) + y
      \height(u_n) + z}
  \end{equation*}
  The lemma follows from the fact that the function
  $n \mapsto n^2/(n \times (x + y) + z)$ has no upper bound.
\end{proof}

We are now ready for the proof that each multi-scale coverable picture
has uniform frequencies. We section it into three technical
Lemmas:~\ref{lemma:2d-cases}, \ref{lemma:2d-upper-class}
and~\ref{lemma:2d-lower-class}.

\begin{lemma} \label{lemma:2d-cases}
  Let $\bw$ be a multi-scale coverable picture. Suppose that $u$ is a
  block of $\bw$ and $(f_u(B_n))_{n \in \bigN}$ does not converge.
  Then there exists:
  \begin{itemize}
  \item a real number $\varepsilon > 0$;
  \item a real number $t \in [0;1]$;
  \item an infinite set $\cv$ of blocks of $\bw$;
  \item an infinite set $\cq$ of covers of $\bw$;
  \end{itemize}
  such that either $f_u(v) < t - \varepsilon < t < f_u(q)$ for all
  $q \in \cq, v \in \cv$, or $f_u(v) > t + \varepsilon > t > f_u(q)$
  for all $q \in \cq, v \in \cv$. Moreover, $L_{\geq K}(\bw) \cap \cv$
  and $L_{\geq K}(\bw) \cap \cq$ are non-empty, for all $K$.
\end{lemma}
\begin{proof}
  Observe that $(f_u(B_n))$ takes its values in $[0;1]$, a compact
  set. Since it does not converge, it has two subsequences converging
  to two different limits, say $\ell_1$ and $\ell_2$ with
  $\ell_1 < \ell_2$. Set $t = (\ell_1 + \ell_2)/2$ and
  $\varepsilon = (\ell_2 - \ell_1)/4$. Define:
  \begin{align*}
    B^- &= \{v \; | \; \text{ $v$ is a block of $\bw$ and } f_u(v) < t \} \\
    B^+ &= \{v \; | \; \text{ $v$ is a block of $\bw$ and } f_u(v) \geq t \}
  \end{align*}

  Remark that $\{B^-, B^+\}$ is a partition of the set of blocks of
  $\bw$ and that both $B^-$ and $B^+$ are infinite (thanks to
  existence of the subsequences). Moreover, $B^-$ and $B^+$ even
  contain squares of arbitrarily large sizes (still thanks to the
  subsequences).

  By pigeonhole principle, there is either infinitely many covers in
  $B^-$ or infinitely many covers in $B^+$. Suppose there are
  infinitely many covers in $B^-$ and call $\cq$ the set of these
  covers. Then we have to set $\cv$. By our previous remarks, there is
  an infinite subsequence of $(B_n)$, call it $B_{\alpha(n)}$, such
  that $f_u(B_{\alpha(n)})$ converges to $\ell_2 = t + 2\varepsilon$.
  By definition of convergence, there exists $N$ such that for all
  $n > N$, we have $f_u(B_{\alpha(n)}) > t + \varepsilon$. Let
  $\cv = \{ B_{\alpha(n)} | n > N \}$.

  Symmetrically, suppose there are infinitely many covers in $B^+$ and
  call $\cq$ the set of these covers. To set $\cv$, observe that there
  is an infinite sequence of blocks with arbitrarily large widths and
  heights whose frequencies converge to $t - 2\varepsilon$. By
  definition of convergence, there is an infinite sequence of blocks
  with arbitrarily large widths and heights whose frequencies are less
  than $t - \varepsilon$. Call the set of images of this sequence
  $\cv$.
\end{proof}

\begin{lemma} \label{lemma:2d-upper-class}
  Let $\bw$ denote a multi-scale coverable picture and $u$ a block of
  $\bw$ without frequency. Suppose there exists some $t \in [0;1]$ and
  some $\varepsilon > 0$ such that:
  \begin{equation*}
    \forall K \in \bigN, \;
    \exists v \in L_{\geq K}(\bw) \text{ and }
    f_u(v) \leq t - \varepsilon
  \end{equation*}
  and that:
  \begin{equation*}
    \forall K \in \bigN, \;
    \exists q \in L_{\geq K}(\bw) \text{ and q is cover of $\bw$ and }
    f_u(q) > t
  \end{equation*}
  Then we have:
  \begin{equation*}
    \forall K' \in \bigN, \;
    \exists v' \in L_{\geq K'}(\bw) \text{ and }
    f_u(v') \leq t - \frac{11}{10}\varepsilon
  \end{equation*}
\end{lemma}
\begin{proof}
  Let $K'$ be some integer; let us find an appropriate block $v'$. The
  proof is in five steps. First we give a summary of the steps.
  \begin{enumerate}
  \item Use Lemma~\ref{lemma:2d-exists-big-things} to choose $q \in
    \cq$ and $v \in \cv$ ``big enough'' for the following steps to
    work well.
  \item Let $\beta$ be a block such that $v \subseteq \beta$ and that
    $\width(\beta)$ is a multiple of $2\width(q)$ and $\height(q)$ is
    a multiple of $2\height(q)$. By hypothesis, $v$ has a ``low''
    frequency; we check that $\beta$ has a ``low'' frequency as well.
  \item Cut $\beta$ into blocks of equal size
    $2\width(q) \times 2\height(q)$. Since the frequency of $\beta$ is
    ``low'', one of the small blocks (call it $b$) must have a ``low''
    frequency as well.
  \item As $b$ has dimensions $2\width(q) \times 2\height(q)$, it must
    contain a full occurrence of $q$. Cut $b$ in two parts: $q$ and
    $m$. By hypothesis, $q$ has a ``high'' frequency, while $b$ as a
    whole has a ``low'' frequency. Therefore, $m$ must have a ``very
    low'' frequency. However, $m$ is not a square, so we extract a
    ``big enough'' square with a ``very low'' frequency.
  \item Do the final calculations to check that everything before is
    correct.
  \end{enumerate}

  \paragraph{Step one}
  Let $q$ be a cover of $\bw$ with $f_u(q) > t$, and let
  $B = 2K' \times (\width(q) + \height(q))$. By hypothesis, we can
  choose such covers with arbitrarily large widths and heights, thus
  use Lemma~\ref{lemma:2d-exists-big-things} to choose $q$
  such that:
  \begin{align}
    \width(q) & > 10K' \label{eq:2d-width-condition} \\
    \height(q) & > 10K' \label{eq:2d-height-condition}
  \end{align}
  \begin{equation*}
    \frac{|q|}{B} > max(10, \frac{1}{\frac{t}{t - \varepsilon/10} - 1})
  \end{equation*}
  Therefore, $q$ satisfies the following relations:
  \begin{align}
    t \times |q| &>
                   (t - \frac{\varepsilon}{10})(|q| + B)
                   \label{eq:2d-step4-helper} \\
    |q| & > 10 \times B \label{eq:2d-step5-helper}
  \end{align}
  We will use Equations~\eqref{eq:2d-width-condition},
  \eqref{eq:2d-height-condition}, \eqref{eq:2d-step4-helper}
  and~\eqref{eq:2d-step5-helper} in Step four.

  Now let $v$ be a block of $\bw$ with $f_u(v) \leq t - \varepsilon$,
  and let
  $A = 2\width(q) \times \height(v) + 2\height(q) \times \width(v) +
  4|q|$.
  Use Lemma~\ref{lemma:2d-exists-big-things} to choose $v$ such that:
  \begin{equation}
    \label{eq:2d-step2-helper} %
    \frac{|v|}{A} > \frac{10}{\varepsilon}
  \end{equation}
  We will use Equation~\eqref{eq:2d-step2-helper} in Step two.
  
  \paragraph{Step two}
  Let $\beta$ be one smallest block of $\bw$ such that:
  \begin{itemize}
  \item $v$ occurs in the bottom left-hand corner of $\beta$;
  \item $\width(\beta) = n \times 2\width(q)$ for some $n \in \bigN$;
  \item $\height(\beta) = m \times 2\height(q)$ for some
    $m \in \bigN$.
  \end{itemize}
  Since $\beta$ is minimal, we have
  $\width(\beta) - \width(v) < 2\width(q)$ and
  $\height(\beta) - \height(v) < 2\height(q)$ (see
  Figure~\ref{fig:anatomy-beta}). So we have
  $|v| \leq |\beta| \leq |v| + A$ (on Figure~\ref{fig:anatomy-beta},
  $A$ is the maximal size of the gray area). In particular,
  $|\beta|_u \leq |v|_u + A$ by Lemma~\ref{lemma:2d-fact}, , hence (by
  $f_u(v) < t-\varepsilon$ and Equation~\eqref{eq:2d-step2-helper}):
  \begin{align*}
    f_u(\beta)
    & \leq f_u(v) \times \frac{|v|}{|\beta|} + \frac{A}{|\beta|} \\
    & \leq f_u(v) + \frac{A}{|v|} \\
    & \leq t - \varepsilon + \frac{A}{|v|} \\
    & \leq t - \varepsilon + \frac{\varepsilon}{10} \\
    & = t - \frac{9\varepsilon}{10}
  \end{align*}

  So we have:
  \begin{equation}
    \label{eq:2d-step2} %
    f_u(\beta) \leq t - \frac{9\varepsilon}{10}
  \end{equation}

  \begin{figure}[h]
    \centering
    \includegraphics{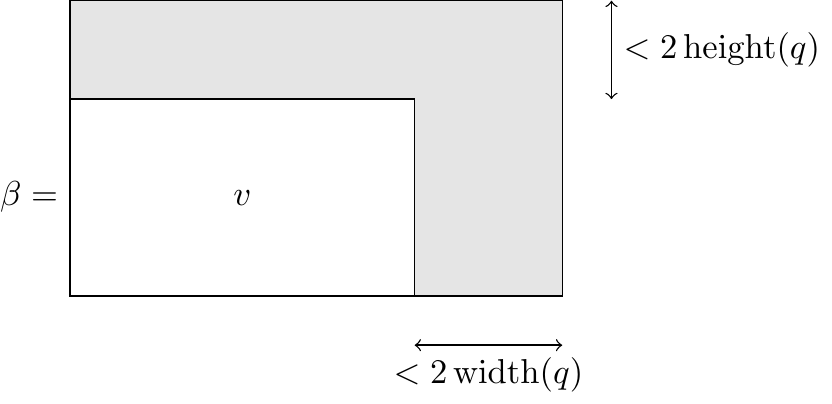}
    \caption{Anatomy of $\beta$.}
    \label{fig:anatomy-beta}
  \end{figure}

  \paragraph{Step three}
  The width and height of $\beta$ are multiples of twice the width and
  height of $q$, respectively. So we can cut $\beta$ into blocks of
  size $2 \width(q) \times 2 \height(q)$. Call $b_{i,j}$ those blocks,
  for $1 \leq i \leq n$ and $1 \leq j \leq m$.
  
  Since $f_u(\beta) \leq t - \frac{9\varepsilon}{10}$
  (Equation~\eqref{eq:2d-step2}), there exists some $i$ and $j$ such
  that $f_u(b_{i,j}) \leq t - \frac{9\varepsilon}{10}$. Suppose not.
  Then, by Lemma~\ref{lemma:2d-fact} applied several times:
  \begin{equation*}
    f_u(\beta) \geq \sum_{1 \leq i \leq n} \sum_{1 \leq j \leq m}
    \frac{|b_{i,j}|}{|\beta|} f_u(b_{i,j}) > (t -
    \frac{9\varepsilon}{10})(\sum_i \sum_j \frac{|b_{i,j}|}{|\beta|}) =
    t - \frac{9\varepsilon}{10} \geq f_u(\beta)
  \end{equation*}
  a contradiction (one inequality is strict). From now on, note $i$
  and $j$ some integers such that:
  $f_u(b_{i,j}) \leq t - \frac{9\varepsilon}{10}$.
  
  \paragraph{Step four}
  De\-com\-po\-se $b_{i,j}$ into five parts: an oc\-cur\-ren\-ce of
  $q$ and four blocks $m_1, m_2, m_3, m_4$ as on
  Figure~\ref{fig:anatomy}.

  \begin{figure}[h]
    \centering
    \includegraphics{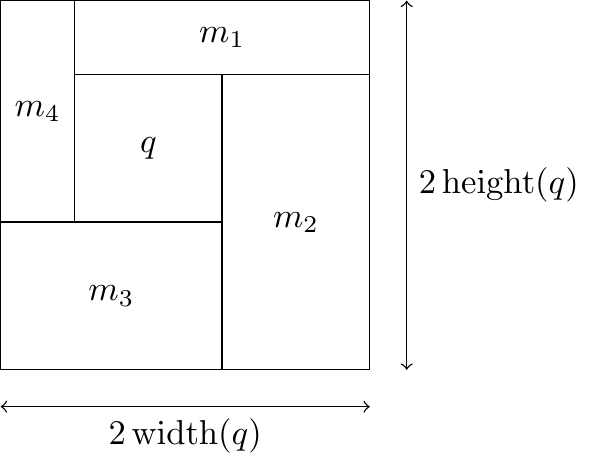}
    \caption{Anatomy of $b_{i,j}$.}
    \label{fig:anatomy}
  \end{figure}

  By definition of the $m_i$'s and
  Equations~\eqref{eq:2d-width-condition}
  and~\eqref{eq:2d-height-condition}, we have:
  \begin{align*}
    \width(m_1) \geq \width(q) &> 10K' \\
    \width(m_3) \geq \width(q) &> 10K' \\
    \width(m_2) + \width(m_4) = \width(q) & > 10K' \\
    \height(m_2) \geq \height(q) &> 10K' \\
    \height(m_4) \geq \height(q) &> 10K' \\
    \height(m_1) + \height(m_3) = \height(q) & > 10K'
  \end{align*}
  Let $m_0$ denote the empty block. We might have either
  $\width(m_2) < K'$ or $\width(m_4) < K'$, but not both. Let
  $\sigma \in \{2,4\}$ such that $\width(m_\sigma) < K'$; by default,
  let $\sigma = 0$. Likewise, we might have $\height(m_1) < K'$ or
  $\height(m_3) < K'$, but not both. Let $\tau \in \{1,3\}$ such that
  $\height(m_\tau) < K'$; by default, let $\tau = 0$. Now set
  $q' = q \cup m_\sigma \cup m_\tau$.
  
  Observe that
  $|q'| = |q| + |m_\sigma| + |m_\tau| \leq |q| + 2 \width(q) \times K'
  + 2 \height(q) \times K'$.
  So we have by Equation~\eqref{eq:2d-step5-helper}:
  \begin{equation*}
    |q| \leq |q'| \leq |q| + B \leq \frac{11}{10}|q|
  \end{equation*}
  and as $|q|/|b_{i,j}| = 1/4$:
  \begin{equation}
    \label{eq:2d-q'-size}
    \frac{1}{4} \leq \frac{|q'|}{|b_{i,j}|} \leq \frac{11}{40}
  \end{equation}
  Moreover, we have $|q'|_u \geq |q|_u$ (as $q'$ contains $q$).
  Therefore, by $f_u(q) > t$ and Equation~\eqref{eq:2d-step4-helper}:
  \begin{equation}
    \label{eq:2d-step4}
    f_u(q') \geq f_u(q) \times \frac{|q|}{|q'|} > t \frac{|q|}{|q|
      + B} > t - \frac{\varepsilon}{10}
  \end{equation}

  \paragraph{Step five}
  We decompose $b_{i,j}$ into $q', m'_1, m'_2, m'_3, m'_4$ as follows.
  If $\width(m_i) > K'$ and $\height(m_i) > K'$, then $m'_i = m_i$;
  otherwise, $m'_i$ is the empty block. Without loss of generality,
  suppose that $m'_1$ is not empty and is minimal for $f_u$ among the
  non-empty $m'_i$'s. Let $M = |m'_1| + |m'_2| + |m'_3| + |m'_4|$.
  Then we have, by Lemma~\ref{lemma:2d-fact} applied $4$ times and
  minimality of $m'_1$ for $f_u$:
  \begin{align}
    |b_{i,j}| & = |q'| +  M \label{eq:2d-bqm-complete} \\
    f_u(b_{i,j}) & \geq \frac{M}{|b_{i,j}|} \times f_u(m'_1) +
                   \frac{|q'|}{|b_{i,j}|} \times f_u(q')
  \end{align}

  \noindent
  Keep these relations in mind and recall that we have from previous
  steps that:
  \begin{align*}
    f_u(b_{i,j})
    & \leq t - \frac{9\varepsilon}{10}
    && \text{ (Step three) } \\
    f_u(q')
    & > t - \frac{\varepsilon}{10}
    && \text{ (Equation~\eqref{eq:2d-step4}) } \\
    \frac{1}{4}
    & \leq \frac{|q'|}{|b_{i,j}|} \leq \frac{11}{40}
    && \text{ (Equation~\eqref{eq:2d-q'-size}) } \\
    \frac{29}{40}
    &\leq \frac{M}{|b_{i,j}|} \leq \frac{3}{4}
    && \text{ (Equations~\eqref{eq:2d-q'-size}
       and~\eqref{eq:2d-bqm-complete}) }
  \end{align*}
  
  Now, assume by contradiction that
  $f_u(m'_1) \geq t - \frac{11\varepsilon}{10}$ and recall that
  $\varepsilon > 0$. Then we have:
  \begin{align*}
    t - \frac{9\varepsilon}{10} \geq f_u(b_{i,j})
    & \geq \frac{M}{|b_{i,j}|}
      \times (t - \frac{11\varepsilon}{10}) +
      \frac{|q'|}{|b_{i,j}|} \times (t - \frac{\varepsilon}{10}) \\
    & \geq t\times\frac{M + |q'|}{|b_{i,j}|} -
      \varepsilon(\frac{3}{4} \times \frac{11}{10} +
      \frac{11}{40} \times \frac{1}{10}) \\
    & \geq t - \frac{341\varepsilon}{400} \gg t -
      \frac{360\varepsilon}{400} = t - \frac{9\varepsilon}{10}
  \end{align*}
  a contradiction. Hence we get
  $f_u(m'_1) < t - \frac{11\varepsilon}{10}$. Set $v' = m'_1$ and the
  lemma is proved.
\end{proof}

Now we need a very similar lemma, but with ``bigger and bigger''
frequencies instead of ``smaller and smaller''. However, before this,
we need an additional technical fact about frequencies.

\begin{lemma}
  \label{lemma:2dbis-upper-bound} %
  Let $u$ denote a finite block. Let $v$ be a finite block which can
  be decomposed into disjoint blocks $v_1, \dots, v_n$ such that
  $\width(v_i) > \width (u)$ and $\height(v_i) > \height(u)$ for all
  $1 \leq i < n$. Then we have:
  \begin{equation*}
    f_u(v) \leq \sum_{i=1}^{n} f_u(v_i) \times \frac{|v_i|}{|v|} +
    \frac{\width(v_i)\times\height(u) + \height(v_i)\times\width(u)}{|v|}
  \end{equation*}
\end{lemma}
\begin{proof}
  Each occurrence of $u$ in $v$ has its bottom left-hand corner in
  some $v_i$. It is either entirely contained in this $v_i$, or it
  overlaps it. There are respectively $|v_i|_u$ and at most
  $\width(v_i)\times\height(u) + \height(v_i)\times\width(u)$ of
  these. Divide by $|v|$ to obtain the result.
\end{proof}

The next lemma is almost identical to
Lemma~\ref{lemma:2d-upper-class}, except that we get bigger and bigger
frequencies instead of lower and lower frequencies. Since the proofs
are very similar (we reverse inequalities and adapt everything to make
it work), we only highlight the differences with the proof of
Lemma~\ref{lemma:2d-upper-class}.

\begin{lemma} \label{lemma:2d-lower-class} %
  Let $\bw$ denote a multi-scale coverable picture and and $u$ a block
  of $\bw$ without frequency. Suppose there exists some $t \in [0;1]$
  and some $\varepsilon > 0$ such that:
  \begin{equation*}
    \forall K \in \bigN,
    \; \exists v \in L_{\geq K}(\bw) \text{ and } f_u(v) \geq t + \varepsilon
  \end{equation*}
  and that:
  \begin{equation*}
    \forall K \in \bigN, \;
    \exists q \in L_{\geq K}(\bw) \text{ cover of $\bw$ and } f_u(q) < t
  \end{equation*}
  Then we have:
  \begin{equation*}
    \forall K' \in \bigN, \;
    \exists v' \in L_{\geq K'}(\bw) \text{ and }
    f_u(v') \geq t + \frac{11}{10}\varepsilon
  \end{equation*}
\end{lemma}
\begin{proof}
  The proof is very similar to the proof of
  Lemma~\ref{lemma:2d-upper-class}. The five big steps are the same.
  We basically have to reverse each inequality, and replace
  $t - x\varepsilon$ with $t + x\varepsilon$ for all values of $x$.
  We use upper bounds instead of lower bounds, which sometimes
  slightly changes the details of the calcuations.
  
  \paragraph{Step one}
  Let $B = 2K' \times (\width(q) + \height(q))$,
  $C = 2\width(q)\height(u) + 2\height(q)\width(u)$ and
  $D = 10(\width(q) + \height(q)) \times |u|$. Use
  Lemma~\ref{lemma:2d-exists-big-things} to choose $q$ such that:
  \begin{align*}
    \width(q) >& 10K' \\
    \height(q) >& 10K' \\
    \frac{|q|}{C} >& \frac{1000}{\varepsilon} \\
    \frac{|q|}{B} >& \frac{\varepsilon}{10} \\
    \frac{|q|}{D} >& \frac{100}{\varepsilon}
  \end{align*}
  Therefore $q$ satisfies the following relations:
  \begin{align}
    \frac{C}{|q|} <& \frac{\varepsilon}{1000}
                      \label{eq:2dbis-step3-helper} \\
    \frac{B}{|q|} <& \frac{10}{\varepsilon}
                              \label{eq:2dbis-step4-helper} \\
    \frac{D}{4|q|} <& \frac{\varepsilon}{400}
                      \label{eq:2dbis-step5-helper}
  \end{align}
  Equations~\eqref{eq:2dbis-step3-helper},
  \eqref{eq:2dbis-step4-helper} and~\eqref{eq:2dbis-step5-helper} will
  be used in steps three, four and five, respectively.

  Let
  $A = 2\width(q) \times \height(v) + 2\height(q) \times \width(v) +
  4|q|$.
  Use Lemma~\ref{lemma:2d-exists-big-things} again to choose $v$ such
  that we have:
  \begin{align*}
    \width(v) & > 2\width(q) \\
    \height(v) & > 2\height(q) \\
    \frac{|v|}{A}
    & > \frac{t + 9\varepsilon/10}{\varepsilon/10}
  \end{align*}
  Therefore, $v$ satisfies the following relation:
  \begin{align}
    (t+\varepsilon) \times |v|
    >& (t + \frac{9\varepsilon}{10})(|v| + A)
       \label{eq:2dbis-step2-helper}
  \end{align}
  Equation~\eqref{eq:2dbis-step2-helper} will be used in step two.
  Observe that $A < 3|v|$.

  \paragraph{Step two}
  Let $\beta$ be one smallest block of $\bw$ such that:
  \begin{itemize}
  \item $v$ occurs in the bottom left-hand corner of $\beta$;
  \item $\width(\beta) = 2 n \width(q)$ for some $n \in \bigN$;
  \item $\height(\beta) = 2 m \height(q)$ for some $m \in \bigN$.
  \end{itemize}
  As in proof of Lemma~\ref{lemma:2d-upper-class}, $|\beta| \leq
  |v|+A$ and so (recall that $A < 3|v|$) we have $|\beta| < 4|v|$.

  As $v \subseteq \beta$, we have $|\beta|_u \geq |v|_u$. Hence, by
  $f_u(v) \geq t + \varepsilon$ and
  Equation~\eqref{eq:2dbis-step2-helper}:
  \begin{equation}
    \label{eq:2dbis-step2} %
    f_u(\beta) \geq f_u(v) \times \frac{|v|}{|\beta|} \geq (t + \varepsilon)
    \times \frac{|v|}{|v| + A} > t + \frac{9\varepsilon}{10}
  \end{equation}

  \paragraph{Step three}
  Decompose $\beta$ into blocks $b_{i,j}$, such that
  $\width(b_{i,j}) = 2\width(q)$ and $\height(b_{i,j}) = 2\height(q)$,
  for $1 \leq i \leq n$ and $1 \leq j \leq m$. In particular,
  $nm = |\beta| / 4|q|$. Since
  $f_u(\beta) > t + \frac{9\varepsilon}{10}$, there exists $i,j$ such
  that $f(b_{i,j}) > t + \frac{89\varepsilon}{100}$. Suppose not.
  Then, by Lemma~\ref{lemma:2dbis-upper-bound}:
  \begin{align*}
    f_u(\beta) 
    \leq & \sum_{i = 1}^n \sum_{j = 1}^m
           f_u(b_{i,j}) \times \frac{|b_{i,j}|}{|\beta|} +
           \frac{C}{|\beta|} \\
    \leq & \sum_{i = 1}^n \sum_{j = 1}^m
           (t + \frac{89\varepsilon}{100}) \times \frac{1}{nm} + 
           \frac{C}{|v|} \\
    \leq & t + \frac{89\varepsilon}{100} + \sum_{i = 1}^n \sum_{j = 1}^m
           \frac{C}{|v|} \\
    \leq & t + \frac{89\varepsilon}{100} +
           \frac{|\beta|}{4|q|} \times \frac{C}{|v|} \\
    \leq & t + \frac{89\varepsilon}{100} +
           \frac{4|v|}{4|q|} \times \frac{C}{|v|} \\
    < & t + \frac{9\varepsilon}{10}
       && \text{(Equation~\eqref{eq:2dbis-step3-helper})}
  \end{align*}
  a contradiction, as $f_u(\beta) > t + \frac{9\varepsilon}{10}$. From
  now on, let $i$ and $j$ denote some integers such that:
  \begin{equation}
    \label{eq:2dbis-step3} %
    f_u(b_{i,j}) > t + \frac{89\varepsilon}{100}
  \end{equation}
  
  \paragraph{Step four}
  Decompose $b_{i,j}$ into $q, m_1, m_2, m_3, m_4$ and define $q'$ as
  in the proof of Lemma~\ref{lemma:2d-upper-class}. We have
  $|q'| \geq |q|$. Then, by Lemma~\ref{lemma:2d-fact}, $f_u(q) < t$ and
  Equation~\eqref{eq:2dbis-step4-helper}:
  \begin{equation}
    \label{eq:2dbis-step4} %
    f_u(q') < f_u(q) \times \frac{|q|}{|q'|} + \frac{B}{|q'|} < t
    \times \frac{|q|}{|q'|} + \frac{B}{|q|} < t + \frac{\varepsilon}{10}
  \end{equation}
  Moreover, Equation~\eqref{eq:2d-q'-size} is still valid.

  \paragraph{Step five} Decompose $b_{i,j}$ into
  $q', m'_1, m'_2, m'_3, m'_4$ as in Lemma~\ref{lemma:2d-upper-class}.
  Let $M = |m'_1| + |m'_2| + |m'_3| + |m'_4|$. We have:
  \begin{equation}
    \label{eq:2dbis-mi-sizes} %
    |b_{i,j}| = 4|q| = |q'| + M
  \end{equation}
  Without loss of generality, suppose $m'_1$ is not empty and is
  maximal for $f_u$ among the $m'_i$'s.

  Now we argue that:
  \paragraph{Fact}
    \begin{equation*}
      |b_{i,j}|_u \leq |q'|_u +
      |m'_1|_u + |m'_2|_u + |m'_3|_u + |m'_4|_u + D
    \end{equation*}

  \begin{proof}
    The previous equation means that $D$ is an upper bound on the
    number of occurrences of $u$ in $b_{i,j}$ which overlap over
    several components of $b_{i,j}$.

    First, view $b_{i,j}$ as decomposed over
    $\{q, m_1, m_2, m_3, m_4\}$ as on Figure~\ref{fig:anatomy}. Then
    the total length of horizontal frontiers is
    $\width(m_4) + \width(m_2) + 2\width(q)$, and the total length
    of vertical frontiers is
    $\height(m_1) + \height(m_3) + 2\height(q)$. Hence, with this
    decomposition, the number of occurrences of $u$ overlapping over
    several components is bounded by
    \begin{align*}
      & \height(u)(\width(m_4) + \width(m_2) + 2\width(q)) + \\
      & \width(u)(\height(m_1) + \height(m_3) + 2\height(q)) \\
      =
      & \height(u) \times 3\width(q) + \width(u) \times 3\height(q)
    \end{align*}
    Re\-call that the decomposition $\{q', m'_1, m'_2, m'_3, m'_4\}$
    is $\{q, m_1, m_2, m_3, m_4\}$ where some components have been
    merged, and others turned to the empty set. Hence the sum of
    lengths of frontiers is shorter. Hence, the number of overlapping
    occurrences of $u$ is smaller. So the given bound is also a bound
    for the number of overlapping occurrences of $u$ for $b_{i,j}$
    decomposed as $\{q', m'_1, m'_2, m'_3, m'_4\}$.
  \end{proof}

  To end the proof of Lemma~\ref{lemma:2d-upper-class}, recall that
  $M = |m'_1| + |m'_2| + |m'_3| + |m'_4|$. We can deduce from the
  previous fact that:
  \begin{equation}
    \label{eq:2dbis-final-eq} %
    f_u(b_{i,j}) \leq f_u(m'_1) \times
    \frac{M}{|b_{i,j}|} + f_u(q') \times \frac{|q'|}{|b_{i,j}|} + \frac{D}{|b_{i,j}|}
  \end{equation}
  
  Keep Equations~\eqref{eq:2dbis-mi-sizes}
  and~\eqref{eq:2dbis-final-eq} in mind, and recall from previous
  steps that:
  \begin{align*}
    f_u(b_{i,j}) & > t + \frac{89\varepsilon}{100}
    && (\text{Equation~\eqref{eq:2dbis-step3}}) \\
    f_u(q') & < t + \frac{\varepsilon}{10}
    && (\text{Equation~\eqref{eq:2dbis-step4}}) \\
    \frac{1}{4} & \leq \frac{|q'|}{|b_{i,j}|} \leq \frac{11}{40}
    && (\text{Equation~\eqref{eq:2d-q'-size}}) \\
    \frac{29}{40} & \leq \frac{M}{|b_{i,j}|} \leq \frac{3}{4}
    && (\text{Equation~\eqref{eq:2d-q'-size} and~\eqref{eq:2d-bqm-complete}}) \\
    \frac{D}{|b_{i,j}|} &< \frac{\varepsilon}{400}
    && (\text{Equation~\eqref{eq:2dbis-step5-helper}})
  \end{align*}
  Finally assume by contradiction that
  $f_u(m'_1) \leq t + \frac{11\varepsilon}{10}$ and recall that
  $\varepsilon > 0$. Then:
  \begin{align*}
    t + \frac{89\varepsilon}{100} < f_u(b_{i,j})
    & < (t + \frac{11\varepsilon}{10})
      \times \frac{M}{|b_{i,j}|} +
      (t + \frac{\varepsilon}{10}) \times \frac{|q'|}{|b_{i,j}|} +
      \frac{D}{|b_{i,j}|} \\
    & < t \times \frac{M + |q'|}{|b_{i,j}|} +
      \varepsilon(\frac{11}{10} \times \frac{3}{4} +
      \frac{1}{10} \times \frac{11}{40} + \frac{1}{400}) \\
    & < t + \frac{342\varepsilon}{400} \ll t +
      \frac{356\varepsilon}{400} = t + \frac{89\varepsilon}{100}
  \end{align*}
  Hence we get $f_u(m'_1) > t + \frac{11\varepsilon}{10}$. Set
  $v' = m'_1$ and the lemma is proved.
\end{proof}

We are now ready for the main proof.

\begin{proof}[Proof of Theorem~\ref{theorem:2d-freqs}]
  Suppose that $\bw$ is a multi-scale coverable picture and that $u$
  is a block of $\bw$ without frequencies. By
  Lemma~\ref{lemma:2d-cases} there exist real numbers $t$ and
  $\varepsilon$, an infinite set of blocks $\cv$ and an infinite set
  of covers $\cq$ such that either $f(v) > t + \varepsilon > t > f(q)$
  for all $q \in \cq$, $v \in \cv$, or
  $f(v) < t < t - \varepsilon < f(q)$ for all $q \in \cq$,
  $v \in \cv$. Moreover, $\cv \cap L_{\geq K}(\bw)$ and
  $\cq \cap L_{\geq K}(\bw)$ are non-empty for all $K \in \bigN$.
  
  Suppose we have $f(v) < t - \varepsilon < t < f(q)$ for all
  $q \in \cq$, $v \in \cv$. In this situation,
  Lemma~\ref{lemma:2d-upper-class} states that there exists an
  infinite set of blocks $\cv'$ such that
  $f(v') \leq t - \frac{11}{10}\varepsilon$ for all $v' \in \cv'$ (and
  $\cv' \cap L_{\geq K}(\bw)$ is non-empty for all $K \in \bigN$). We
  can apply the same lemma again and get blocks whose frequency is
  $\leq t - (\frac{11}{10})^2\varepsilon$. Then we can apply the lemma
  again and again, until we get
  $f(v') \leq t - (\frac{11}{10})^n\varepsilon < 0$ for some $n$. At
  that point, we get blocks with negative frequencies: a
  contradiction.

  If $f(v) > t+\varepsilon > t > f(q)$ for all $v \in \cv, q \in \cq$,
  the proof follows the same idea, except that we use
  Lemma~\ref{lemma:2d-lower-class}. We get blocks with higher and
  higher frequencies, until we find a block with a frequency bigger
  than $1$: a contradiction again.
\end{proof}

\subsection{Uniform Recurrence}
\label{sec:ms-ur}

Recall that a $\bigZ^2$-word $\bw$ is \emph{uniformly recurrent} when
all its blocks occur infinitely often with bounded gaps. In
$\bigN$-words, multi-scale coverability implies uniform recurrence.
However, this result does not hold for $\bigZ$-words. Here is an
example of a $\bigZ$-word which is multi-scale coverable, but not
uniformly recurrent:
\begin{equation*}
  ^\omega(ab)a(ab)^\omega = \dots{}babababa\,a\,babababa\dots
\end{equation*}
Any word matching the $aba(ba)^*$ regular expression is a covering
pattern of this word. However, the pattern $aa$ only occur once, hence
it is not uniformly recurrent.

It is easy to generalize these results to pictures.
\begin{proposition}
  \label{proposition:n2-ur}
  Any multi-scale, $\bigN^2$-word $\bw$ is uniformly recurrent.
\end{proposition}
\begin{proof}
  This is an adaptation of the proof from~\cite{MarcusMonteil2006Arxiv}.
  Consider a rectangle $r$ occurring in $\bw$. Since $\bw$ has
  arbitrarily large covering patterns and all these patterns occur at
  the origin, one of these patterns contains $r$ entirely. Hence $r$
  occurs whenever the covering patterns occurs, and the latter occurs
  infinitely many times with bounded gaps.
\end{proof}

Now let us see an example of coverable $\bigZ^2$-words which is not
uniformly recurrent (or even recurrent). Consider $q =
\begin{smallmatrix}
  b & b & a \\
  b & b & b \\
  a & b & b
\end{smallmatrix}$
and the word displayed on Figure~\ref{fig:non-ur}. The central block:
\begin{equation*}
  \begin{matrix}
    b & b & b & b & a \\
    b & b & b & b & b \\
    b & b & a & b & b \\
    b & b & b & b & b \\
    a & b & b & b & b \\
  \end{matrix}
\end{equation*}
occurs only once, hence this picture is not uniformly recurrent.

To get uniform recurrence back, we propose a notion of \emph{strong}
multi-scale coverable. A word (or a picture) is \emph{strongly}
multi-scale coverable if and only if any factor (or block) occuring
also occurs in a cover. Observe that multi-scale coverability and its
strong counterpart are equivalent on $\bigN$-words.

\begin{proposition}
  An infinite picture $\bw$ is strongly multi-scale coverable if and
  only if it is multi-scale coverable and uniformly recurrent.
\end{proposition}
\begin{proof}
  Strong multi-scale coverability implies multi-scale coverability:
  each occurring block must occur within a cover, and there are
  arbitrarily large (in width and height) blocks, so there must be
  arbitrarily large (in width and height) covers.
  
  Strong multi-scale coverability also implies uniform recurrence,
  almost by definition: any block occurs in a cover, which in turn
  occurs infinitely often with bounded gaps.

  Finally, multi-scale coverability and uniform recurrence imply
  strong multi-scale coverability. Indeed, let $\bw$ be a multi-scale
  coverable picture which is also uniformly recurrent, and let $B$ be
  a block of $\bw$. Since $\bw$ is uniformly recurrent, there exists
  $n \in \bigN$ such that any $n \times n$-block of $\bw$ contain an
  occurrence of $B$. By definition of multi-scale coverability, there
  are covers which are bigger than $2n \times 2n$ (both in width and
  in height); such covers must contain an occurrence of $B$.
\end{proof}

\section{Conclusion}
\label{sec:conclu}

Coverability is a local rule. Our aim was to determine whether this
notion enforces some global properties on covered words. Although this
is not the case in one dimension, Theorem~\ref{theorem:qp-ent0} shows
that, under some natural hypotheses on the cover (natural in the sense
they take into account the two dimensions), this enforcement is
possible in two dimensions. However, many other questions have to be
considered to better understand the power of coverability, especially
when considering that covers are rectangular blocks.

Our approach could be linked to considerations from dynamical systems
and tilings. For instance, a natural question is: are self-similarity
and multi-scale coverability linked? Our study already states that, as
in the one-dimensional case, multi-scale coverability implies other
properties, such as existence of frequencies and zero topological
entropy. A difference with the one-dimensional case is that multiscale
coverability does not impy uniform recurrence.
Proposition~\ref{proposition:n2-ur} and discussion before explain that
this difference does not come directly from the change of dimension
but much more from the fact that we consider the full plane $\bigZ^2$
instead of the quarter of plane $\bigN^2$.

One could ask why considering $\bigZ^2$-words as generalization of
$\bigN$-words. This stems from the fact that $\bigZ^2$-words are much
more relevant in the area of tilings. Let us observe that most
results, such as Proposition~\ref{proposition:coverable-nothing},
Proposition~\ref{proposition:msqp-ent0} and
Theorem~\ref{theorem:2d-freqs}, can be directly adapted to
$\bigN^2$-words. However, it is much more difficult for
Theorem~\ref{theorem:qp-ent0} and Theorem~\ref{theorem:construct}, as
the condition on border is not adequate for $\bigN^2$-words.

\section*{Acknowledgements}

We acknowledge Ana{\"{e}}l Grandjean for his help with the proof of
Theorem~\ref{theorem:2d-freqs} and many interesting discussions during
the preparation of this paper.

\section*{References}

\bibliographystyle{abbrv}
\bibliography{biblio}{}

\begin{thebibliography}{10}

\bibitem{ApostolicoEhrenfeucht1993TCS}
A.~Apostolico and A.~Ehrenfeucht.
\newblock Efficient detection of quasiperiodicities in strings.
\newblock {\em Theor. Comput. Sci.}, 119(2):247--265, 1993.

\bibitem{Brudno1982Trudy}
A.~A. {Brudno}.
\newblock {Entropy and complexity of the trajectories of a dynamical system}.
\newblock {\em {Tr. Mosk. Mat. O.-va}}, 44:124--149, 1982.

\bibitem{Cassaigne1999DLT}
J.~Cassaigne.
\newblock Subword complexity and periodicity in two or more dimensions.
\newblock In G.~Rozenberg and W.~Thomas, editors, {\em Developments in Language
  Theory, Foundations, Applications, and Perspectives, Aachen, Germany, 6--9
  July 1999}, pages 14--21. World Scientific, 1999.

\bibitem{CrochemoreIliopoulosKorda1998Algorithmica}
M.~Crochemore, C.~S. Iliopoulos, and M.~Korda.
\newblock Two-dimensional prefix string matching and covering on square
  matrices.
\newblock {\em Algorithmica}, 20(4):353--373, 1998.

\bibitem{DurandLevinShen2005MathInt}
B.~Durand, L.~Levin, and A.~Shen.
\newblock Local rules and global order, or aperiodic tilings.
\newblock {\em Math. Intell.}, 27(1):64--68, 2005.

\bibitem{CANT7}
S.~Ferenczi and T.~Monteil.
\newblock Infinite words with uniform frequencies, and invariant measures.
\newblock In V.~Berth{\'{e}} and M.~Rigo, editors, {\em Combinatorics, Automata
  and Number Theory}. Cambridge University Press, Cambridge, 2010.

\bibitem{GamardRichomme2015Lata}
G.~Gamard and G.~Richomme.
\newblock Coverability in two dimensions.
\newblock In A.~H. Dediu, E.~Formenti, C.~Mart{\'{\i}}n{-}Vide, and B.~Truthe,
  editors, {\em Language and Automata Theory and Applications - 9th
  International Conference, {LATA} 2015, Nice, France, March 2-6, 2015,
  Proceedings}, volume 8977 of {\em Lecture Notes in Computer Science}, pages
  402--413. Springer, 2015.

\bibitem{GlenLeveRichomme2008TCS}
A.~Glen, F.~Lev{\'{e}}, and G.~Richomme.
\newblock Quasiperiodic and {L}yndon episturmian words.
\newblock {\em Theor. Comput. Sci.}, 409(3):578--600, 2008.

\bibitem{LeveRichomme2004BEATCS}
F.~Lev{\'{e}} and G.~Richomme.
\newblock Quasiperiodic infinite words: Some answers (column: Formal language
  theory).
\newblock {\em Bulletin of the {EATCS}}, 84:128--138, 2004.

\bibitem{LeveRichomme2007TCS}
F.~Lev{\'{e}} and G.~Richomme.
\newblock Quasiperiodic {S}turmian words and morphisms.
\newblock {\em Theor. Comput. Sci.}, 372(1):15--25, 2007.

\bibitem{Lothaire1997}
M.~Lothaire.
\newblock {\em Combinatorics on Words}.
\newblock Cambridge Mathematical Library. Cambridge University Press, 1997.

\bibitem{Marcus2004BEATCS}
S.~Marcus.
\newblock Quasiperiodic infinite words (columns: Formal language theory).
\newblock {\em Bulletin of the {EATCS}}, 82:170--174, 2004.

\bibitem{MarcusMonteil2006Arxiv}
T.~Monteil and S.~Marcus.
\newblock Quasiperiodic infinite words: multi-scale case and dynamical
  properties.
\newblock {\em CoRR}, abs/math/0603354, 2006.

\bibitem{PolleyStaiger2010DFCS}
R.~Polley and L.~Staiger.
\newblock The maximal subword complexity of quasiperiodic infinite words.
\newblock In I.~McQuillan and G.~Pighizzini, editors, {\em Proceedings Twelfth
  Annual Workshop on Descriptional Complexity of Formal Systems, {DCFS} 2010,
  Saskatoon, Canada, 8-10th August 2010.}, volume~31 of {\em {EPTCS}}, pages
  169--176, 2010.

\end{thebibliography}

\end{document}